\documentclass[runningheads]{llncs}

\newif\ifproc
\procfalse 

\usepackage{graphicx}
\usepackage{hyperref}
\usepackage{amssymb, amsmath}
\graphicspath{{figures/}}
\usepackage[misc,geometry]{ifsym} 

\renewcommand{\paragraph}[1]{\smallskip\noindent\textbf{\textsf{#1}}}
\usepackage{csquotes}
\usepackage[font=small]{subcaption}
\usepackage[font=small]{caption}
\usepackage[disable]{todonotes}
\usepackage{booktabs}
\usepackage[inline]{enumitem}
\setlist[enumerate]{nosep}
\usepackage{hyperref}
\usepackage{timetravel}
\setCounters{theorem,lemma,corollary}
\usepackage{wrapfig}
\usepackage{xspace}

\let\doendproof\endproof
\renewcommand\endproof{~\hfill$\qed$\doendproof}
\spnewtheorem*{sketch}{Proof sketch}{\itshape}{\rmfamily}

\newcounter{casecounter}
\newcounter{subcasecounter}
\newcounter{subsubcasecounter}
\makeatletter
\newcommand{\ccase}[1]{%
  \stepcounter{casecounter}%
  \setcounter{subcasecounter}{0}%
  \setcounter{subsubcasecounter}{0}%
  \protected@write \@auxout {}{\string \newlabel {#1}{{\thecasecounter}{\thepage}{\thecasecounter}{#1}{}} }%
  \hypertarget{#1}{\noindent\textbf{Case \thecasecounter.}}
}
\newcommand{\subcase}[1]{%
  \stepcounter{subcasecounter}%
  \setcounter{subsubcasecounter}{0}%
  \protected@write \@auxout {}{\string \newlabel {#1}{{\thecasecounter.\thesubcasecounter}{\thepage}{\thecasecounter.\thesubcasecounter}{#1}{}} }%
  \hypertarget{#1}{\noindent\textbf{Case \thecasecounter.\thesubcasecounter.}}
}
\newcommand{\subsubcase}[1]{%
  \stepcounter{subsubcasecounter}%
  \protected@write \@auxout {}{\string \newlabel {#1}{{\thecasecounter.\thesubcasecounter.\thesubsubcasecounter}{\thepage}{\thecasecounter.\thesubcasecounter.\thesubsubcasecounter}{#1}{}} }%
  \hypertarget{#1}{\noindent\textbf{Case \thecasecounter.\thesubcasecounter.\thesubsubcasecounter.}}
}
\makeatother

\newcommand{\V}{\mathcal{V}}

\title{Drawing Subcubic 1-Planar Graphs\\with Few Bends, Few Slopes, and Large Angles}

\author{Philipp Kindermann\inst{1}\ifproc$^{\textrm{(\Letter)}}$\orcidID{0000-0001-5764-7719}\fi
\and Fabrizio Montecchiani\inst{2}
\and Lena Schlipf\inst{3}
\and Andr\'e~Schulz\inst{3}}
\institute{University of Waterloo, Canada,
	\email{philipp.kindermann@uwaterloo.ca}
	\and
	Universit\`a degli Studi di Perugia, Italy,
	\email{fabrizio.montecchiani@unipg.it}
	\and
  FernUniversit\"at in Hagen, Germany, 
  \email{firstname.lastname@fernuni-hagen.de}}

\authorrunning{P.~Kindermann, F. Montecchiani, L. Schlipf, and A. Schulz}
\titlerunning{Subcubic 1-Planar Graphs with Few Bends, Few Slopes, Large Angles} 

\begin{document}
\maketitle
\begin{abstract}
We show that the $1$-planar slope number of $3$-connected cubic $1$-planar graphs is at most $4$ when edges are drawn as polygonal curves with at most $1$ bend each. 
This bound is obtained by drawings whose vertex and crossing resolution is at least $\pi/4$.  
On the other hand, if the embedding is fixed, then there is a  $3$-connected cubic $1$-planar graph 
that needs $3$ slopes when drawn with at most $1$ bend per edge.
We also show that $2$ slopes always suffice for $1$-planar drawings of subcubic $1$-planar graphs with at most $2$ bends per edge. 
This bound is obtained with vertex resolution $\pi/2$ and the drawing is RAC (crossing resolution $\pi/2$).
Finally, we prove lower bounds for the slope number of straight-line $1$-planar drawings in terms of
number of vertices and  maximum degree.
\end{abstract}


\section{Introduction}\label{sec:intro}
A graph is \emph{$1$-planar} if it can be drawn in the plane such that each 
edge is crossed at most once. The notion of $1$-planarity naturally extends 
planarity and received considerable attention since its first introduction by 
Ringel in 1965~\cite{R65}, as witnessed by recent surveys~\cite{DBLP:journals/corr/abs-1804-07257,DBLP:journals/csr/KobourovLM17}. 
Despite the efforts made in the study of $1$-planar graphs, only few results are known 
concerning their geometric representations (see, 
e.g.,~\cite{DBLP:conf/gd/AlamBK13,DBLP:journals/tcs/BekosDLMM17,DBLP:journals/comgeo/Brandenburg18,DBLP:journals/algorithmica/GiacomoDELMMW18}). 
In this paper, we study the 
existence of $1$-planar drawings that simultaneously satisfy the following 
properties: edges are polylines using few bends and few distinct slopes for 
their segments, edge crossings occur at large angles, and pairs of edges incident 
to the same vertex form large angles. For example, Fig.~\ref{fig:intro} shows 
a $1$-bend drawing of a $1$-planar graph (i.e., a drawing in which each edge 
is a polyline with at most one bend) using $4$ distinct slopes, such that 
edge crossings form angles at least $\pi/4$, and the angles 
formed by edges incident to the same vertex are at least $\pi/4$. In 
what follows, we briefly recall known results concerning the problems of 
computing polyline drawings with few bends and few slopes or with few bends 
and large~angles. 

\paragraph{Related work.} The \emph{$k$-bend (planar) slope number} of a 
(planar) graph $G$ with maximum vertex degree $\Delta$ is the minimum number 
of distinct edge slopes needed to compute a (planar) drawing of $G$ such that 
each edge is a polyline with at most $k$ bends. When $k=0$, this 
parameter is simply known as the \emph{(planar) slope number} of $G$. 
Clearly, if $G$ has maximum vertex degree $\Delta$, at least $\lceil \Delta/2 \rceil$ 
slopes are needed for any~$k$.
While there exist non-planar graphs with $\Delta \geq 5$ whose slope number 
is unbounded with respect to $\Delta$~\cite{DBLP:journals/combinatorics/BaratMW06,DBLP:journals/combinatorics/PachP06}, 
Keszegh et al.~\cite{DBLP:journals/siamdm/KeszeghPP13} proved that the planar slope number 
is bounded by $2^{O(\Delta)}$. Several authors improved this bound for 
subfamilies of planar graphs (see, e.g.,~\cite{DBLP:journals/gc/JelinekJKLTV13,DBLP:journals/comgeo/KnauerMW14,DBLP:conf/gd/LenhartLMN13}).  

Concerning $k$-bend drawings, Angelini et al.~\cite{DBLP:conf/compgeom/AngeliniBLM17} proved 
that the $1$-bend planar slope number is at most $\Delta-1$, while Keszegh et 
al.~\cite{DBLP:journals/siamdm/KeszeghPP13} proved that the $2$-bend planar 
slope number is $\lceil \Delta/2 \rceil$ (which is tight). Special attention
has been paid in the literature to the slope number of \emph{(sub)cubic} 
graphs, i.e., graphs having vertex degree (at most) 3. Mukkamala and P{\'a}lv{\"o}lgyi 
showed that the four slopes $\{0, \frac{\pi}{4}, \frac{\pi}{2}, \frac{3\pi}{4}\}$ 
suffice for every cubic graph~\cite{DBLP:conf/gd/MukkamalaP11}. 
For planar graphs, Kant and independently Dujmovi\'{c} et al. proved that 
cubic $3$-connected planar graphs have planar slope number $3$ disregarding 
the slopes of three edges on the outer face~\cite{DBLP:journals/comgeo/DujmovicESW07,k-hgd-wg91}, 
while Di Giacomo et al.~\cite{DBLP:journals/tcs/GiacomoLM18} proved that the planar slope number of 
subcubic planar graphs is $4$. We also remark that the slope number problem 
is related to orthogonal drawings, which are planar and with slopes $\{0, \frac{\pi}{2}\}$~\cite{orthoChapterHandbook}, 
and with octilinear drawings, which are planar and with slopes $\{0, \frac{\pi}{4}, \frac{\pi}{2}, \frac{3\pi}{4}\}$~\cite{DBLP:journals/jgaa/BekosG0015}. 
All planar graphs with $\Delta\le4$ (except 
the octahedron) admit 2-bend orthogonal drawings~\cite{bk-bhogd-cgta98,lms-la2be-DAM98}, 
and planar graphs admit octilinear drawings without bends if 
$\Delta\le 3$~\cite{DBLP:journals/tcs/GiacomoLM18,k-hgd-wg91}, with 1 bend if 
$\Delta\le5$~\cite{DBLP:journals/jgaa/BekosG0015}, and with 2 bends if 
$\Delta\le 8$~\cite{DBLP:journals/siamdm/KeszeghPP13}.

Of particular interest for us is the \emph{$k$-bend $1$-planar slope number} 
of $1$-planar graphs, i.e., the minimum number of distinct edge slopes needed 
to compute a $1$-planar drawing of a $1$-planar graph such that each edge is 
a polyline with at most $k \ge 0$ bends. 
Di Giacomo et al.~\cite{DBLP:journals/jgaa/GiacomoLM15} proved an $O(\Delta)$ upper bound for 
the $1$-planar slope number ($k=0$) of outer $1$-planar graphs, i.e., graphs 
that can be drawn $1$-planar with all vertices on the external boundary. 

Finally,   the \emph{vertex resolution} and the \emph{crossing 
resolution} of a drawing are defined as the minimum angle between two 
consecutive segments incident to the same vertex or crossing, respectively (see, e.g.,~\cite{DBLP:journals/jgaa/DuncanK03,DBLP:journals/siamcomp/FormannHHKLSWW93,DBLP:journals/siamdm/MalitzP94}). 
A drawing is \emph{RAC} (\emph{right-angle crossing}) if its crossing resolution is $\pi/2$. Eades 
and Liotta proved that $1$-planar graphs may not have straight-line RAC 
drawings~\cite{DBLP:journals/dam/EadesL13}, while Chaplick et al.~\cite{DBLP:conf/ewcg/Chaplick18} 
and Bekos et al.~\cite{DBLP:journals/tcs/BekosDLMM17} 
proved that every $1$-planar graph has a $1$-bend~RAC~drawing
that preserves the embedding.

\paragraph{Our contribution.} We prove upper and lower bounds on the $k$-bend $1$-planar slope number of $1$-planar graphs, when $k \in \{0,1,2\}$. Our results are based on  techniques that lead to drawings with large vertex and crossing resolution.

In Section~\ref{sec:3con1bend}, we prove that every $3$-connected cubic $1$-planar graph admits a $1$-bend $1$-planar drawing that uses at most~4 distinct slopes and has both vertex and 
crossing resolution~$\pi/4$. 
In Section~\ref{sec:2bends}, we show that every subcubic $1$-planar graph admits a $2$-bend $1$-planar drawing that uses at most~2 distinct slopes and has both vertex and crossing resolution~$\pi/2$. These bounds on the number of slopes and on the vertex/crossing resolution are clearly worst-case optimal.
In Section~\ref{subsec:lower-1bend}, we give a $3$-connected cubic $1$-plane graph for which any embedding-preserving $1$-bend drawing uses at least~3 distinct slopes. 
The lower bound holds even if we are allowed to change the outer face.
In Section~\ref{subsec:lower-straight}, we present $2$-connected subcubic 1-plane graphs with $n$ vertices such that any embedding-preserving straight-line drawing uses~$\Omega(n)$ distinct slopes,
and $3$-connected $1$-plane graphs with maximum degree $\Delta\geq 3$ such that any embedding-preserving straight-line drawing uses at least $9(\Delta-1)$ distinct slopes, which implies that at least $18$ slopes are needed if $\Delta=3$.

\noindent Preliminaries can be found in Section~\ref{sec:prelim}, while  
open problems are in Section~\ref{sec:open}. 


\section{Preliminaries}\label{sec:prelim}

We only consider \emph{simple} graphs with neither self-loops nor multiple 
edges. A \emph{drawing} $\Gamma$ of a graph $G$  maps each vertex of $G$ to a 
point of the plane and  each edge to a simple open Jordan curve between its 
endpoints. We always refer to \emph{simple} drawings where two edges can 
share at most one point, which is either a common endpoint or a proper intersection. 
A drawing 
divides the plane into topologically connected regions, called \emph{faces}; 
the infinite region is called the \emph{outer face}. For a planar (i.e., crossing-free) drawing, the 
boundary of a face consists of vertices and edges, while for a non-planar 
drawing the boundary of a face may also contain crossings and parts of edges.
An \emph{embedding} of a graph $G$ is an equivalence class 
of drawings of~$G$ that define the same set of faces and the same outer face. 
A \emph{($1$-)plane graph} is a graph with 
a fixed ($1$-)planar embedding. Given a $1$-plane graph $G$, the 
\emph{planarization}~$G^*$ of $G$ is the plane graph obtained by replacing each 
crossing of $G$ with a \emph{dummy vertex}. To avoid confusion, the vertices 
of $G^*$ that are not dummy are called \emph{real}. Moreover, we call 
\emph{fragments} the edges of $G^*$ that are incident to a dummy vertex. 
The next lemma will be used in the following and can be of independent 
interest, as it extends a similar result by Fabrici and Madaras~\cite{DBLP:journals/dm/FabriciM07}. 
\ifproc
The proof is given in the full version~\cite{fullversion}.
\else
The proof is given in Appendix~\ref{app:prelim}.
\fi

\wormhole{crossingMinimal}
\newcommand{\crossingMinimalText}{%
  Let~$G=(V,E)$ be a $1$-plane graph and let $G^*$ be its planarization. 
  We can re-embed~$G$ such that each edge is still crossed at most once and 
  \begin{enumerate*}[label=(\roman*)]
    \item no cutvertex of~$G^*$ is a dummy vertex, and
    \item if~$G$ is 3-connected, then $G^*$ is 3-connected.
  \end{enumerate*}}

\begin{lemma}\label{le:crossing-minimal}
  \crossingMinimalText
\end{lemma}

A drawing $\Gamma$ is \emph{straight-line} if all its edges are mapped to 
segments, or it is \emph{$k$-bend} if each edge is mapped to a chain of 
segments with at most $k>0$ bends. The \emph{slope} of an edge segment 
of $\Gamma$ is the slope of the line containing this segment. For 
convenience, we measure the slopes by their angle with respect to the $x$-axis. 
Let $\mathcal S=\{\alpha_1,\dots,\alpha_t\}$ be a set of $t$ distinct 
slopes.  The \emph{slope number} of a $k$-bend drawing $\Gamma$ is the number 
of distinct slopes used for the edge segments of $\Gamma$. An edge segment of 
$\Gamma$ uses the \emph{north (N) port} (\emph{south (S) port}) of a vertex~$v$ 
if it has slope $\pi/2$ and $v$ is its bottommost (topmost) endpoint. We 
can define analogously the  \emph{west (W)} and \emph{east (E)} ports with 
respect to the slope $0$, the \emph{north-west (NW)} and  \emph{south-east (SE)} 
ports with respect to slope $3\pi/4$, and the \emph{south-west (SW)} and 
\emph{north-east (NE)} ports with respect to slope $\pi/4$. Any such port 
is \emph{free} for $v$ if there is no edge that attaches to $v$ by using it.

We will use a decomposition technique called \emph{canonical ordering}~\cite{DBLP:journals/algorithmica/Kant96}. 
Let $G=(V,E)$ be a $3$-connected plane graph. Let $\delta = \{\V_1,\dots,\V_K\}$ be an ordered partition of $V$, that is, $\V_1 \cup \dots \cup \V_K = V$ and $\V_i \cap \V_j = \emptyset$ for $i \neq j$. Let~$G_i$ be the subgraph of $G$ induced by $\V_1 \cup \dots \cup \V_i$ and denote by $C_i$ the outer face of $G_i$. 
The partition $\delta$ is a canonical ordering of $G$ if:
(i) $\V_1=\{v_1,v_2\}$, where $v_1$ and $v_2$ lie on the outer face of $G$ and $(v_1,v_2) \in E$.  
(ii) $\V_K = \{v_n\}$, where $v_n$ lies on the outer face of $G$, $(v_1,v_n) \in E$.
(iii) Each $C_i$ ($i > 1$) is a cycle containing $(v_1,v_2)$. 
(iv) Each $G_i$ is $2$-connected and internally $3$-connected, that is, removing any two interior vertices of $G_i$ does not disconnect it. 
(v) For each $i \in \{2, \dots, K-1\}$, one of the following conditions holds:
(a) $\V_i$ is a \emph{singleton} $v^i$ that lies on $C_i$ and has at least one neighbor in $G \setminus G_i$; 
(b) $\V_i$ is a \emph{chain} $\{v^i_1,\dots, v^i_l\}$,  both $v^i_1$ and~$v^i_l$ have exactly one neighbor each in $C_{i-1}$, and $v^i_2, \ldots, v^i_{l-1}$ have no neighbor in~$C_{i-1}$. Since $G$ is $3$-connected, each $v^i_j$ has at least one neighbor in $G \setminus G_i$.

Let $v$ be a vertex in $\V_i$, then its neighbors in $G_{i-1}$ (if $G_{i-1}$ 
exists) are called the \emph{predecessors} of $v$, while its neighbors in 
$G\setminus G_{i}$ (if $G_{i+1}$ exists) are called the \emph{successors} of $v$. In 
particular, every singleton has at least two predecessors and at least one 
successor, while every vertex in a chain has either zero or one predecessor 
and at least one successor. Kant~\cite{DBLP:journals/algorithmica/Kant96} 
proved that a canonical ordering of $G$ always exists and can be computed in $O(n)$ 
time; the technique in~\cite{DBLP:journals/algorithmica/Kant96} is such 
that one can arbitrarily choose two adjacent vertices $u$ and $w$ on the 
outer face so that $u=v_1$ and $w=v_2$ in the computed canonical ordering. 

An $n$-vertex \emph{planar $st$-graph} $G=(V,E)$ 
is a plane acyclic directed graph with a single source $s$ and a single sink~$t$, 
both on the outer face~\cite{DBLP:journals/tcs/BattistaT88}. 
An \emph{$st$-ordering} 
of $G$  is a numbering $\sigma: V 
\rightarrow \{1,2,\dots,n\}$ such that for each edge $(u,v) \in E$, it holds $
\sigma(u) < \sigma(v)$ (thus $\sigma(s)=1$ and $\sigma(t)=n$). 
For an  $st$-graph, an $st$-ordering can be computed in $O(n)$ time 
(see, e.g.,~\cite{DBLP:books/daglib/0023376}) and every biconnected undirected 
graph can be oriented to become a planar $st$-graph (also in linear time).


\section{1-bend Drawings of $3$-connected cubic $1$-planar graphs}\label{sec:3con1bend}

Let $G$ be a $3$-connected $1$-plane cubic graph, and let $G^*$ be its planarization. 
We can assume that $G^*$ is $3$-connected (else we can re-embed~$G$ by Lemma~\ref{le:crossing-minimal}).  
We choose as outer face of $G$ a face containing an edge $(v_1,v_2)$ whose vertices are both real (see Fig.~\ref{fig:g}). 
Such a face exists: If~$G$ has $n$ vertices, then~$G^*$ has fewer 
than $3n/4$ dummy vertices because~$G$ is subcubic.
Hence we find a face in $G^*$ with more real than dummy vertices and hence with two consecutive real vertices.
Let $\delta = \{\V_1,\dots,\V_K\}$ be a canonical ordering of $G^*$, let $G_i$ be the graph obtained by adding the first $i$ sets of $\delta$ 
and let~$C_i$ be the outer face of $G_i$.

Note that a real vertex $v$ of $G_i$ can have at most one successor~$w$ in some 
set $\V_j$ with $j>i$. 
We call~$w$ an \emph{L-successor} (resp., \emph{R-successor}) of~$v$ if~$v$ is 
the leftmost (resp., rightmost) neighbor of $\V_j$ on $C_{i}$.
Similarly, a dummy vertex $x$ of $G_i$ can have at most two successors in 
some sets $\V_j$ and $\V_l$ with $l \ge j > i$.
In both cases, a vertex~$v$ of~$G_i$ having a successor in some 
set $\V_j$ with $j>i$ is called \emph{attachable}.
We call~$v$ \emph{L-attachable} (resp., \emph{R-attachable}) 
if~$v$ is attachable and has
no L-successor (resp., R-successor) in~$G_i$. 
We will draw an upward edge at~$u$ with slope $\pi/4$ (resp., $3\pi/4$) only if 
it is L-attachable (resp., R-attachable).

\begin{figure}[t]
  \centering
  \begin{subfigure}[b]{.23\textwidth}
    \centering
    \includegraphics[page=1]{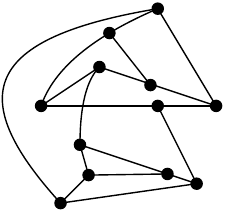}
    \caption{$G$}
    \label{fig:g}
  \end{subfigure}
  \hfil
  \begin{subfigure}[b]{.23\textwidth}
    \centering
    \includegraphics[page=2]{figures/example}
    \caption{$\delta$}
    \label{fig:canord}
  \end{subfigure}
  \hfil
  \begin{subfigure}[b]{.23\textwidth}
    \centering
    \includegraphics[page=3]{figures/example}
    \caption{$uv$-cut}
    \label{fig:cut}
  \end{subfigure}
  \hfil
  \begin{subfigure}[b]{.23\textwidth}
    \centering
    \includegraphics{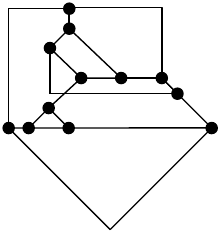}
    \caption{}
    \label{fig:intro}
  \end{subfigure}
  \caption{(a) A $3$-connected $1$-plane cubic graph $G$; 
  (b) a canonical ordering $\delta$ of the planarization $G^*$ of $G$---the real (dummy) vertices are black points (white squares); 
  (c) the edges crossed by the dashed line are a $uv$-cut of $G_5$ with respect to $(u,w)$---the two components have a yellow and a blue background, respectively;
  (d) a $1$-bend $1$-planar drawing with $4$ slopes of~$G$}
\end{figure}

Let $u$ and $v$ be two vertices of $C_i$, for $i>1$. Denote by $P_i(u,v)$ the path of $C_i$ having $u$ and $v$ as endpoints and that does not contain $(v_1,v_2)$. 
Vertices~$u$ and $v$ are \emph{consecutive} if they are both attachable and if $P_i(u,v)$ does not contain any other attachable vertex. 
Given two consecutive vertices $u$ and $v$ of~$C_i$ and an edge $e$ of~$C_i$, a \emph{$uv$-cut of $G_i$ with respect to $e$} is a set of edges of $G_i$ that contains both $e$ and $(v_1,v_2)$ and whose removal disconnects $G_i$ into two components, one containing $u$ and one containing $v$ (see Fig.~\ref{fig:cut}).
We say that~$u$ and~$v$ are \emph{L-consecutive} (resp., \emph{R-consecutive}) if they are consecutive,~$u$ lies to the left (resp., \emph{right}) of~$v$ on~$C_i$, and~$u$ is L-attachable (resp., R-attachable).

We construct an embedding-preserving drawing $\Gamma_i$ of $G_i$, for $i=2,\dots,K$, by adding one by one the sets of $\delta$.  
A drawing $\Gamma_i$ of $G_i$ is \emph{valid}, if: 

\begin{enumerate}[label={\bf P\arabic*}]

\item\label{P1} It uses only slopes in the set $\{0,\frac{\pi}{4},\frac{\pi}{2},\frac{3\pi}{4}\}$;

\item\label{P2} It is a $1$-bend drawing such that the union of any two edge fragments that correspond to the same edge in $G$ is drawn with (at most) one bend in total. 

\end{enumerate}
 
\noindent A valid drawing $\Gamma_K$ of $G_K$ will coincide with the desired drawing of $G$, after replacing dummy vertices with crossing points. 

\paragraph{Construction of $\Gamma_2$.} 
We begin by showing how to draw $G_2$.
We distinguish two cases, based on whether $\V_2$ is a singleton or a chain, as illustrated in Fig.~\ref{fig:g2}.
 
\paragraph{Construction of $\Gamma_i$, for $2 < i < K$.} 
We now show how to compute a valid drawing of $G_i$, for $i=3,\dots,K-1$, by incrementally adding the sets of $\delta$. 

We aim at constructing a valid drawing $\Gamma_i$ that is also \emph{stretchable}, i.e., that satisfies the following two more properties; see Fig.~\ref{fig:p34}. 
These two properties will be useful to prove Lemma~\ref{le:stretch}, which defines a standard way of stretching a drawing by lengthening horizontal segments.

\begin{enumerate}[label={\bf P\arabic*}]
\setcounter{enumi}{2}
\item\label{P3} The edge $(v_1,v_2)$ is drawn with two segments $s_1$ and $s_2$ that meet at a point~$p$. Segment $s_1$ uses the SE port of $v_1$ and $s_2$ uses the SW port of $v_2$. Also, $p$ is the lowest point of $\Gamma_i$, and no other point of $\Gamma_i$ is contained by the two lines that contain $s_1$ and $s_2$.

\begin{figure}[t]
  \begin{minipage}[b]{.7\textwidth}
    \centering
    \begin{subfigure}[b]{.32\textwidth}
      \centering
      \includegraphics[page=1]{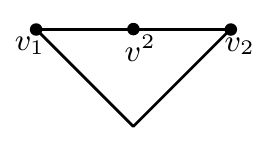}
      \caption{}
      \label{fig:g2-reals}
    \end{subfigure}
    \hfil
    \begin{subfigure}[b]{.32\textwidth}
      \centering
      \includegraphics[page=2]{figures/g2new}
      \caption{}
      \label{fig:g2-dummys}
    \end{subfigure}
    \hfil
    \begin{subfigure}[b]{.32\textwidth}
      \centering
      \includegraphics[page=3]{figures/g2new}
      \caption{}
      \label{fig:g2-chain}
    \end{subfigure}
    \caption{Construction of $\Gamma_2$: (a) $\V_2$ is a real singleton; 
    (b)~$\V_2$ is a dummy singleton; 
    (c) $\V_2$ is a chain.}
    \label{fig:g2}
  \end{minipage}
  \hfil
  \begin{minipage}[b]{.23\textwidth}
    \centering
    \includegraphics[page=4]{figures/example}
    \caption{$\Gamma_i$ is stretchable.}
    \label{fig:p34}
  \end{minipage}
\end{figure}

\item\label{P4}
For every pair of consecutive vertices $u$ and $v$ of $C_i$ with~$u$ left of~$v$
on~$C_i$, it holds that 
\begin{enumerate*}[label=(\alph*)]
	\item If~$u$ is L-attachable (resp., $v$ is R-attachable), then the path $P_i(u,v)$ 
    is such that for each vertical segment $s$ on this path there is a horizontal segment 
    in the subpath before $s$ if $s$ is traversed upwards when going from $u$ to $v$ (resp., from $v$ to $u$); 
  \item if both $u$ and $v$ are real, then~$P_i(u,v)$ contains at least one horizontal segment; and
  \item for every edge~$e$ 
    of $P_i(u,v)$ such that $e$ contains a horizontal segment, there exists a $uv$-cut 
    of $G_i$ with respect to $e$ whose edges all contain a horizontal segment 
    in $\Gamma_i$ except for $(v_1,v_2)$, and such that there exists a $y$-monotone 
    curve that passes through all and only such horizontal segments and 
    $(v_1,v_2)$.
\end{enumerate*}

\end{enumerate}

\wormhole{stretch}
\newcommand{\stretchText}{%
Suppose that $\Gamma_i$ is valid and stretchable, and let $u$ and $v$ be two consecutive vertices of $C_i$.
If~$u$ is L-attachable (resp., $v$ is R-attachable), then
it is possible to modify $\Gamma_i$ such that any half-line with slope $\pi/4$ (resp., $3\pi/4$) that originates at $u$ (resp.,  at $v$) and that intersects the outer face of $\Gamma_i$ does not intersect any edge segment with slope $\pi/2$ of $P_i(u,v)$. Also, the modified drawing is still valid and stretchable.}
\begin{lemma}\label{le:stretch}
  \stretchText
\end{lemma}
\begin{sketch}
Crossings between such half-lines and vertical segments of $P_i(u,v)$ can be solved 
by finding suitable $uv$-cuts and moving everything on
the right/left side of the cut to the right/left.
The full proof is given in 
\ifproc
the full version~\cite{fullversion}.
\else
Appendix~\ref{app:3con1bend}
\fi
\end{sketch}

Let $P$ be a set of ports of a vertex $v$; the \emph{symmetric}  set of 
ports~$P'$ of~$v$ is the set of ports obtained by mirroring $P$ at a 
vertical line through $v$.
We say that $\Gamma_i$ is 
\emph{attachable} if the following two properties also apply. 

\begin{enumerate}[label={\bf P\arabic*}]
\setcounter{enumi}{4}
\item\label{P5} At any attachable real vertex~$v$ of $\Gamma_i$, its N, NW, and NE ports are free. 

\item\label{P6} Let $v$ be an attachable dummy vertex of $\Gamma_i$. 
If $v$ has two successors, there are four possible cases for its two used ports, illustrated with two \emph{solid} edges in Fig.~\ref{fig:case-a}--\subref{fig:case-d}.  
If~$v$ has only one successor not in $\Gamma_i$, there are eight possible cases for its three used ports, illustrated with two solid edges plus one \emph{dashed} or one \emph{dotted} edge in Fig.~\ref{fig:case-a}--\subref{fig:case-e} 
\ifproc
.
\else
(see Fig.~\ref{fig:Cxx} in Appendix~\ref{app:3con1bend}). 
\fi

\begin{figure}[t]
\centering
\begin{minipage}[b]{.19\textwidth}
\centering
\includegraphics[width=\textwidth, page=1]{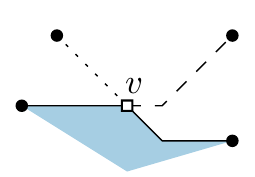}
\subcaption{$C1$}\label{fig:case-a}
\end{minipage}\hfill
\begin{minipage}[b]{.19\textwidth}
\centering
\includegraphics[width=\textwidth, page=2]{figures/cases}
\subcaption{$C2$}\label{fig:case-b}
\end{minipage}
\begin{minipage}[b]{.19\textwidth}
\centering
\includegraphics[width=\textwidth, page=5]{figures/cases}
\subcaption{$C2$ symm.}\label{fig:case-c}
\end{minipage}
\begin{minipage}[b]{.19\textwidth}
\centering
\includegraphics[width=\textwidth, page=3]{figures/cases}
\subcaption{$C3$}\label{fig:case-d}
\end{minipage}
\begin{minipage}[b]{.19\textwidth}
\centering
\includegraphics[width=\textwidth, page=6]{figures/cases}
\subcaption{$C3$ symm.}\label{fig:case-e}
\end{minipage}
\caption{Illustration for \ref{P6}. If $v$ has two successors not in $\Gamma_i$, 
  then the edges connecting~$v$ to its two neighbors in $\Gamma_i$ are solid. If 
  $v$ has one successor in~$\Gamma_i$, then the edge between~$v$ and
  this successor is dashed or dotted.
  \label{fig:cases}}
\end{figure}

\end{enumerate}

\noindent Observe that $\Gamma_2$, besides being valid, is also stretchable and attachable by construction (see also Fig.~\ref{fig:g2}). 
Assume that $G_{i-1}$ admits a valid, stretchable, and attachable drawing $\Gamma_{i-1}$, for some $2 \le i < K-1$; 
we show how to add the next set $\V_i$ of $\delta$ so to obtain a drawing $\Gamma_i$ of $G_i$ that is valid, stretchable and attachable. 
We distinguish between the following cases. 

\ccase{c:singleton} $\V_i$ is a singleton, i.e., $\V_i=\{v^i\}$. Note 
that if $v^i$ is real, it has two neighbors on $C_{i-1}$, while if it is 
dummy, it can have either two or three neighbors on $C_{i-1}$. Let $u_l$ 
and $u_r$ be the first and the last neighbor of $v^i$, respectively, when walking 
along $C_{i-1}$ in clockwise direction from $v_1$. We will call $u_l$ (resp., $u_r$) 
the \emph{leftmost predecessor} (resp., \emph{rightmost predecessor}) of $v^i$.

\subcase{c:singleton-real} Vertex $v^i$ is real. Then, $u_l$ and $u_r$ are its 
only two neighbors in $C_{i-1}$. Each of $u_l$ and $u_r$ can be real or 
dummy. If $u_l$ (resp., $u_r$) is real, we draw $(u_l,v^i)$ (resp., $(u_r,v^i)$) 
with a single segment using the NE port of $u_l$ and the SW port of $v^i$ 
(resp., the NW port of $u_r$ and the SE port of $v^i$). 
If $u_l$ is dummy and 
has two successors not in $\Gamma_{i-1}$, we distinguish between the 
cases of Fig.~\ref{fig:cases} as shown in Fig.~\ref{fig:cases-real-sing}. 
The symmetric configuration of $C3$ is only used for connecting to $u_r$.

\begin{figure}[b]
  \centering
  \begin{subfigure}[b]{.19\textwidth}
    \centering
    \includegraphics[width=\textwidth, page=3]{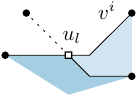}
    \caption{$C1$}\label{fig:case-a-sr}
  \end{subfigure}
  \hfill
  \begin{subfigure}[b]{.19\textwidth}
    \centering
    \includegraphics[width=\textwidth, page=4]{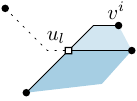}
    \caption{$C2$}\label{fig:case-b-sr}
  \end{subfigure}
  \hfill
  \begin{subfigure}[b]{.19\textwidth}
    \centering
    \includegraphics[width=\textwidth, page=5]{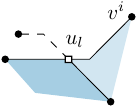}
    \caption{$C2$ symm.}\label{fig:case-bsym-sr}
  \end{subfigure}
  \hfill
  \begin{subfigure}[b]{.19\textwidth}
    \centering
    \includegraphics[width=\textwidth, page=6]{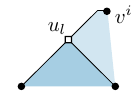}
    \caption{$C3$}\label{fig:case-c-sr}
  \end{subfigure}
  \caption{A real singleton when $u_l$ is dummy with two successors not in $\Gamma_{i-1}$
  \label{fig:cases-real-sing}}
\end{figure}

If $u_l$ is dummy and has one successor not in $\Gamma_{i-1}$, we distinguish 
between the various cases of Fig.~\ref{fig:cases} as indicated in Fig.~\ref{fig:cases-real-sing-2}%
\ifproc
.
\else
(see Fig.~\ref{fig:Cxx} for all cases in Appendix~\ref{app:3con1bend}). 
\fi
Observe that $C1$ requires a local reassignment of 
one port of $u_l$.  The edge $(u_r,v^i)$ is drawn by following a similar case 
analysis%
\ifproc
.
\else
(depicted in Fig.~\ref{fig:Cxx} of Appendix~\ref{app:3con1bend}). 
\fi
Vertex $v^i$ is then placed at the intersection of the lines passing through 
the assigned ports, which always intersect by construction. In particular, 
 the S port is only used when $u_l$ has one successor, but the same 
situation cannot occur when drawing $(u_r,v^i)$.
Otherwise, there is a path of~$C_{i-1}$ from~$u_l$ via its successor~$x$ on $C_{i-1}$ 
to~$u_r$ via its successor~$y$ on~$C_{i-1}$. Note that $x=y$ is possible but $x\neq u_r$.
Since the first edge on this path 
goes from a predecessor to a successor and the last edge goes from a 
successor to a predecessor, there has to be a vertex~$z$ without a successor 
on the path; but then~$u_l$ and~$u_r$ are not consecutive. To avoid crossings 
between  $\Gamma_{i-1}$ and the new edges $(u_l,v^i)$ and $(u_r,v^i)$, we 
apply Lemma~\ref{le:stretch} to suitably stretch the drawing. 
In particular, possible crossings can occur only with vertical edge segments of $P_{i-1}(u_l,u_r)$, 
because when walking along $P_{i-1}(u_l,u_r)$ from $u_l$ to $u_r$ we only encounter a (possibly empty) set of segments with slopes in the range 
$\{3\pi/4,\pi/2,0\}$, followed by a (possibly empty) set of segments with slopes in the range $\{\pi/2,\pi/4,0\}$.

\begin{figure}[t]
  \centering
  \begin{subfigure}[b]{.19\textwidth}
    \centering
    \includegraphics[width=\textwidth, page=3]{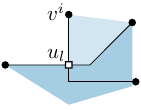}
    \caption{$C1$}
    \label{fig:case-a-sr-2}
  \end{subfigure}
  \hfill
  \begin{subfigure}[b]{.19\textwidth}
    \centering
    \includegraphics[width=\textwidth, page=4]{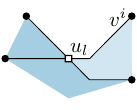}
    \caption{$C1$}
    \label{fig:case-b-sr-2}
  \end{subfigure}
  \begin{subfigure}[b]{.19\textwidth}
    \centering
    \includegraphics[width=\textwidth, page=5]{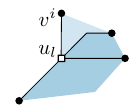}
    \caption{$C2$}
    \label{fig:case-c-sr-2}
  \end{subfigure}
  \hfill
  \begin{subfigure}[b]{.19\textwidth}
    \centering
    \includegraphics[width=\textwidth, page=6]{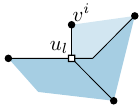}
    \caption{$C2$ symm.}
    \label{fig:case-dsym-sr-2}
  \end{subfigure}
  \hfill
  \begin{subfigure}[b]{.19\textwidth}
    \centering
    \includegraphics[width=\textwidth, page=7]{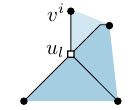}
    \caption{$C3$}
    \label{fig:case-d-sr-2}
  \end{subfigure}
  \caption{Some cases for the addition of a real singleton when $u_l$ is dummy with one successor not in $\Gamma_{i-1}$
  \label{fig:cases-real-sing-2}}
\end{figure}

\begin{figure}[b]
  \centering
  \begin{subfigure}[b]{.24\textwidth}
    \centering
    \includegraphics[width=\textwidth, page=1]{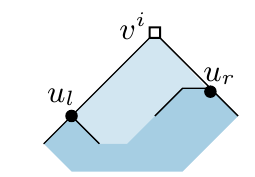}
    \caption{}
    \label{fig:case-a-sd}
  \end{subfigure}
  \hspace{2cm}
  \begin{subfigure}[b]{.24\textwidth}
    \centering
    \includegraphics[width=\textwidth, page=4]{figures/cases-sing-dummy}
    \caption{}
    \label{fig:case-d-sd}
  \end{subfigure}
  \caption{Illustration for the addition of a dummy singleton
  \label{fig:cases-dummy-sing}}
\end{figure}

\subcase{c:singleton-dummy} Vertex $v^i$ is dummy. By $1$-planarity, the two or 
three neighbors of~$v^i$ on $C_{i-1}$ are all real. If $v^i$ 
has two neighbors,  we  draw $(u_l,v^i)$ and $(u_r,v^i)$ as shown in 
Fig.~\ref{fig:case-a-sd}, while if $v^i$ has three neighbors,
 we draw $(u_l,v^i)$ and $(u_r,v^i)$ as shown in Fig.~\ref{fig:case-d-sd}. 
Analogous to the previous 
case, vertex $v^i$ is placed at the intersection of the lines passing 
through the assigned ports, which always intersect by construction, and  
avoiding crossings between  $\Gamma_{i-1}$ and the new edges $(u_l,v^i)$ and $
(u_r,v^i)$  by applying Lemma~\ref{le:stretch}. In particular, if $v^i$ has 
three neighbors on $C_{i-1}$, say $u_l$, $w$, and $u_r$, 
by \ref{P4} there is a horizontal segment between $u_l$ and $w$, 
as well as between $w$ and $u_r$. 
Thus, Lemma~\ref{le:stretch} can be applied not only 
to resolve crossings, but also to find a suitable point where the two lines 
with slopes~$\pi/4$ and~$3\pi/4$ meet along the line with slope 
$\pi/2$ that passes through $w$. 

\ccase{c:chain} $\V_i$ is a chain, i.e., $\V_i=\{v^i_1,v^i_2,\dots,v^i_l\}$. 
We find a point as if we had to place a vertex $v$ whose leftmost 
predecessor is the leftmost predecessor of~$v^i_1$ and whose rightmost 
predecessor is the rightmost predecessor of~$v^i_l$. We then draw the chain 
slightly below this point by using the same technique used to draw $\V_2$.
 Again, Lemma~\ref{le:stretch} can be applied to resolve 
possible crossings.

We formally prove the correctness of our algorithm in 
\ifproc
the full version~\cite{fullversion}.
\else
Appendix~\ref{app:3con1bend}.
\fi

\wormhole{gkvalid}
\newcommand{\gkvalidText}{%
Drawing $\Gamma_{K-1}$ is valid, stretchable, and attachable.}
\begin{lemma}\label{le:gk-1}
  \gkvalidText
\end{lemma}

\begin{figure}[t]
  \centering
  \begin{subfigure}[b]{.32\textwidth}
    \centering
    \includegraphics[page=2]{figures/vn}
    \caption{$v_n$ is dummy}
    \label{fig:vn-dummy}
  \end{subfigure}
  \hfil
  \begin{subfigure}[b]{.32\textwidth}
    \centering
    \includegraphics[page=1]{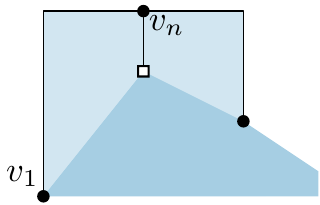}
    \caption{$v_n$ is real}
    \label{fig:vn-real}
  \end{subfigure}
  \caption{Illustration for the addition of $\V_k$
  \label{fig:vn}}
\end{figure}

\paragraph{Construction of $\Gamma_K$.} 
We now show how to add $\V_K=\{v_n\}$ to $\Gamma_{K-1}$ so as to obtain a valid 
drawing of $G_K$, and hence the desired drawing of $G$ after replacing dummy 
vertices with crossing points. Recall that $(v_1,v_n)$ is an edge of $G$ by the
definition of canonical ordering. We distinguish whether $v_n$ is real or 
dummy; the two cases are shown in Fig.~\ref{fig:vn}. Note that if $v_n$ is 
dummy, its four neighbors are all real and hence their N, 
NW, and NE ports are free by \ref{P5}. If $v_n$ is real, it has three neighbors in 
$\Gamma_{K-1}$, $v_1$ is real by construction, and the S port can be
used to attach with a dummy vertex. Finally, since $\Gamma_{K-1}$ is 
attachable, we can use Lemma~\ref{le:stretch} to avoid crossings and to find 
a suitable point to place $v_n$. A complete 
drawing is shown in Fig.~\ref{fig:intro}.

The theorem follows immediately by the choice of the slopes.

\begin{theorem}\label{thm:3con1bend}
	Every $3$-connected cubic $1$-planar graph admits a 
	$1$-bend $1$-planar drawing with at most~4 distinct slopes 
	and angular and crossing resolution~$\pi/4$.
\end{theorem}


\section{2-bend drawings}\label{sec:2bends}

Liu et al.~\cite{lms-la2be-DAM98} presented an algorithm to compute orthogonal
drawings for planar graphs of maximum degree~4 with at most~2 bends per edge
(except the octahedron, which requires~3 bends on one edge). 
We make use of their algorithm for biconnected graphs.
The algorithm chooses two vertices~$s$ and~$t$ and computes an
$st$-ordering of the input graph. Let $V=\{v_1,\ldots,v_n\}$ with $\sigma(v_i)=i$,
$1\le i\le n$. Liu et al. now compute an embedding of~$G$ such that~$v_2$ lies on
the outer face if $\deg(s)=4$ and~$v_{n-1}$ lies on the outer face if $\deg(t)=4$;
such an embedding exists for every graph with maximum degree~4 except the octahedron.

The edges around each vertex~$v_i,1\le i\le n$, are assigned to the four ports as follows.
If~$v_i$ has only one outgoing edge, it uses the N port; if~$v_i$ has two
outgoing edges, they use the N and E port; if~$v_i$ has three
outgoing edges, they use the N, E, and W port; and if~$v_i$
has four outgoing edges, they use all four ports. Symmetrically,
the incoming edges of~$v_i$ use the S, W, E, and N port, in this order.
The edge~$(s,t)$ (if it exists) is assigned to the W port of both~$s$
and~$t$. If~$\deg(s)=4$, the edge $(s,v_2)$ is assigned to the S port of~$s$
(otherwise the port remains free);
if~$\deg(t)=4$, the edge $(t,v_{n-1})$ is assigned to the N port of~$t$
(otherwise the port remains free).
Note that every vertex except~$s$ and~$t$ has at least one incoming and one
outgoing edge; hence, the given embedding of the graph provides a unique assignment 
of edges to ports. 
Finally, they place the vertices bottom-up as prescribed
by the $st$-ordering.
The way an edge is drawn is determined completely by the port assignment,
as depicted in Fig.~\ref{fig:2bend-shapes}.

Let $G=(V,E)$ be a subcubic 1-plane graph. We first re-embed~$G$ according to
Lemma~\ref{le:crossing-minimal}. Let~$G^*$ be the planarization of~$G$
after the re-embedding. Then, all cutvertices of~$G^*$ are real vertices,
and since they have maximum degree~3, there is always a bridge connecting
two 2-connected components.
Let $G_1,\ldots,G_k$ be the 2-connected components of~$G$, 
and let~$G_i^*$ be the planarization of $G_i,1\le i\le k$.
We define the \emph{bridge decomposition tree} $\mathcal{T}$ of $G$ as the graph 
having a node for each component $G_i$ of $G$, and an edge 
$(G_i,G_j)$, for every pair $G_i, G_j$ connected by a bridge in $G$. 
We root~$\mathcal{T}$ in~$G_1$.
For each component $G_i,2\le i\le k$,
let~$u_i$ be the vertex of~$G_i$ connected to the parent of~$G_i$ in~$\mathcal T$
by a bridge and let~$u_1$ be an arbitrary vertex of~$G_1$. 
We will create a drawing~$\Gamma_i$ for each component~$G_i$ with
at most 2 slopes and 2 bends such that~$u_i$ lies on the outer face.

To this end, we first create a drawing~$\Gamma_i^*$ of~$G_i^*$ with the algorithm of Liu et 
al.~\cite{lms-la2be-DAM98} and then modify the drawing. Throughout the
modifications, we will make sure that the following invariants hold for the
drawing~$\Gamma_i^*$.
\begin{enumerate}[label=(I\arabic*)]
	\item $\Gamma_i^*$ is a planar orthogonal drawing of $G_i^*$ and edges
  are drawn as in Fig.~\ref{fig:2bend-shapes};
  \item $u_i$ lies on the outer face of $\Gamma_i^*$ and its N port is free;
  \item every edge is $y$-monotone from its source to its target;
  \item every edge with 2 bends is a C-shape, there are no edges with more bends;
  \item if a C-shape ends in a dummy vertex, it uses only E ports; and
  \item if a C-shape starts in a dummy vertex, it uses only W ports.
\end{enumerate}

\begin{figure}[t]
  \centering
  \includegraphics{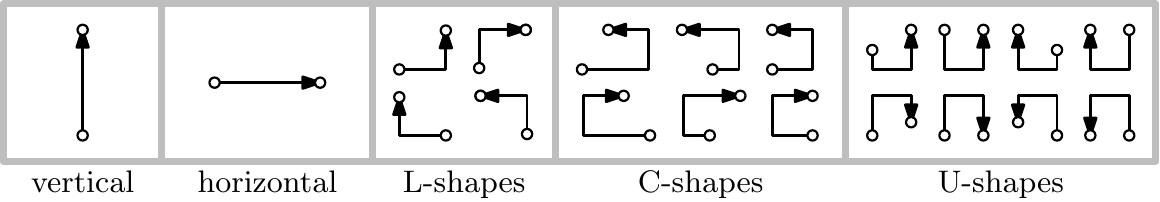}
  \caption{The shapes to draw edges}
  \label{fig:2bend-shapes}
\end{figure}

\wormhole{bendInvariants}
\newcommand{\bendInvariantsText}{%
  Every $G_i^*$ admits a drawing~$\Gamma_i^*$ that satisfies
  invariants (I1)--(I6).}
\begin{lemma}\label{lem:2bend-invariants}
  \bendInvariantsText
\end{lemma}
\begin{sketch}
  We choose~$t=u_i$ and some real vertex~$s$ and use the algorithm by Liu et al.\
  to draw~$G_i^°$. Since~$s$ and~$t$ are real, there are no U-shapes.
  Since no real vertex can have an outgoing edge at its W port or incoming edge
  at its E port, the invariants follow.
  The full proof is given in 
  \ifproc
  the full version~\cite{fullversion}.
  \else
  Appendix~\ref{app:2bends}.
  \fi
\end{sketch}

We now iteratively remove the C-shapes from the drawing while
maintaining the invariants.
We make use of a technique similar to the stretching in Section~\ref{sec:3con1bend}.
We lay an orthogonal $y$-monotone curve~$S$ through our drawing that intersects no
vertices. Then we stretch the drawing by moving~$S$ and all features that lie right of~$S$ 
to the right, and stretching all points on~$S$ to horizontal
segments. After this stretch, in the area between the old and the
new position of~$S$, there are only horizontal segments of edges
that are intersected by~$S$. The same operation can be defined symmetrically 
for an $x$-monotone curve that is moved upwards.

\wormhole{bendBicon}
\newcommand{\bendBiconText}{%
  Every $G_i$ admits an orthogonal 2-bend drawing such that~$u_i$ lies on the 
  outer face and its N port is free.}
\begin{lemma}\label{lem:2bend-bicon}
  \bendBiconText
\end{lemma}
\begin{sketch}
  We start with a drawing~$\Gamma_i^*$ of~$G_i^*$ that satisfies invariants
  (I1)--(I6), which exists by Lemma~\ref{lem:2bend-invariants}.
  By (I2),~$u_i$ lies on the outer face and its N port is free.
  If no dummy vertex in~$\Gamma_i^*$ is incident to a C-shape,  
  by (I4) all edges incident to dummy vertices are drawn with at most~1 bend,
  so the resulting drawing~$\Gamma_i$ of~$G_i$ is an orthogonal 2-bend drawing.
  Otherwise, there is a C-shape between a real vertex~$u$ and
  a dummy vertex~$v$. We show how to eliminate this C-shape without 
  introducing new ones while maintaining all invariants. 
    
  We prove the case that $(u,v)$ is directed from~$u$ to~$v$, so by (I5) it uses
  only E ports; the other case is symmetric.
  We do a case analysis based on which ports at~$u$ are free.
  We show one case here and the rest in 
  \ifproc
  the full version~\cite{fullversion}.
  \else
  Appendix~\ref{app:2bends}.
  \fi
  
  \setcounter{casecounter}{0}
  \ccase{c:bottomleft} The N port at~$u$ is free;
  see Fig.~\ref{fig:2bend-bottomleft}. Create a curve~$S$ as follows: 
  Start at some point~$p$ slightly to the top left of~$u$ and extend it downward 
  to infinity. Extend it from~$p$ to the right until it passes the
  vertical segment of $(u,v)$ and extend it upwards to infinity. Place the
  curve close enough to~$u$ and $(u,v)$ such that no vertex or bend point
  lies between~$S$ and the edges of~$u$ that lie right next to it.
  Then, stretch the drawing by moving~$S$ to the right such that~$u$ is
  placed below the top-right bend point of $(u,v)$. Since~$S$ intersected a
  vertical segment of~$(u,v)$, this changes the edge to be drawn with~4 bends.
  However, now the region between~$u$ and the second bend point of $(u,v)$
  is empty and the N port of~$u$ is free, so we can make an L-shape out
  of $(u,v)$ that uses the N port at~$u$. This does not change the drawing
  style of any edge other than $(u,v)$, so all the invariants are maintained
  and the number of C-shapes is reduced by one.
  \begin{figure}[t]
    \centering
    \subcaptionbox{}{\includegraphics[page=1,scale=0.9]{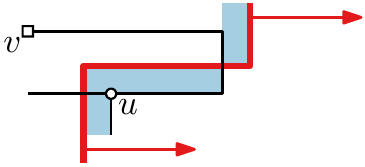}}
    \hfill
    \subcaptionbox{}{\includegraphics[page=2,scale=0.9]{2bend-removeCTiny}}
    \hfill
    \subcaptionbox{}{\includegraphics[page=3,scale=0.9]{2bend-removeCTiny}}
    \caption{Proof of Lemma~\ref{lem:2bend-bicon}, Case~\ref{c:bottomleft}}
    \label{fig:2bend-bottomleft}
  \end{figure}
\end{sketch}

Finally, we combine the drawings~$\Gamma_i$ to a drawing~$\Gamma$ of~$G$.
Recall that every cutvertex is real and two biconnected
components are connected by a bridge. Let~$G_j$ be a child of~$G_i$
in the bridge decomposition tree. We have drawn~$G_j$ with~$u_j$ on the outer
face and a free N port. Let~$v_i$ be the neighbor of~$u_j$ in~$G_i$.
We choose one of its free ports, rotate and scale~$\Gamma_j$ such that it fits
into the face of that port, and connect~$u_j$ and~$v_i$ with a vertical
or horizontal segment. Doing this for every biconnected component
gives an orthogonal 2-bend drawing of~$G$.

\begin{theorem}\label{thm:2bends}
	Every subcubic 1-plane graph admits a 2-bend 1-planar 
	drawing with at most~2 distinct slopes and both angular and crossing 
	resolution~$\pi/2$.
\end{theorem}


\section{Lower bounds for 1-plane graphs}\label{sec:lower}


\subsection{1-bend drawings of subcubic graphs}\label{subsec:lower-1bend}

\begin{theorem}\label{thm:lower-1bend}
	There exists a subcubic 3-connected 1-plane graph such that any 
	embedding-preserving 1-bend drawing uses at least~3 distinct slopes. 
	The lower bound holds even if we are allowed to change the outer face.
\end{theorem}
\begin{proof}
Let $G$ be the $K_4$ with a planar embedding. The outer face is a 3-cycle, which
has to be drawn as a polygon~$\Pi$ with at least four (nonreflex) corners. 
Since we allow only one bend per edge, one of the corners
of~$\Pi$ has to be a vertex of~$G$. The vertex in the interior has to connect to this corner, however,
all of its free ports lie on the outside. Thus, no drawing of~$G$ is possible. 
\end{proof}
\todo{the idea of making the example larger does not work .... need to think about it.}


\subsection{Straight-line drawings}\label{subsec:lower-straight}

The full proofs for this section are given in 
\ifproc
the full version~\cite{fullversion}.
\else
Appendix~\ref{app:lower}.
\fi

\wormhole{lowerStraight}
\newcommand{\lowerStraightTwoRegText}{%
	There exist 2-regular 2-connected 1-plane graphs with $n$ vertices
	such that any embedding-preserving straight-line drawing 
	uses~$\Omega(n)$ distinct slopes.}
\begin{theorem}\label{thm:lower-straight-2reg}
  \lowerStraightTwoRegText
\end{theorem}
\begin{figure}[t]
  \subcaptionbox{Theorem~\ref{thm:lower-straight-2reg} \label{fig:lowerbound-2reg-a-main}}{\includegraphics[page=1]{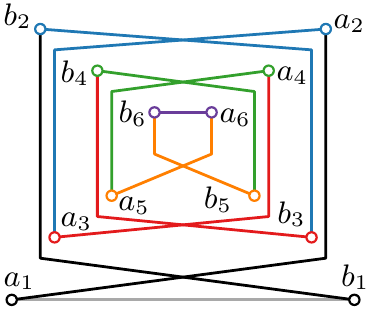}}
  \hfill
  \subcaptionbox{Lemma~\ref{lem:lower-straight-3reg}  \label{fig:lowerbound-3reg-main}}{\includegraphics[page=1,width=.27\textwidth]{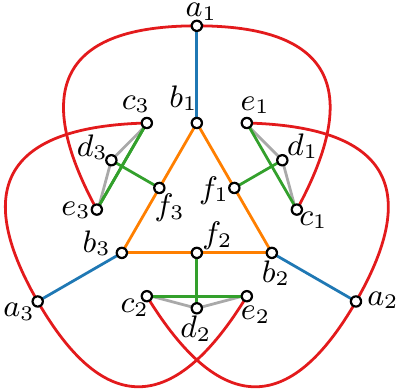}}
  \hfill
   \subcaptionbox{Theorem~\ref{thm:lower-straight-gen} \label{fig:lowerbound-gen-main}}{\includegraphics[page=3,width=.27\textwidth]{lowerbound-3reg}}
  \caption{The constructions for the results of Section~\ref{sec:lower}}
  \label{fig:lowerbound-main}%
\end{figure}
\begin{sketch}
  Let $G_k$ be the graph given by the cycle  $a_1\ldots,a_{k+1},b_{k+1},\ldots,b_1,a_1$
  and the embedding shown in Fig.~\ref{fig:lowerbound-2reg-a-main}. 
  Walking along the path $a_1,\ldots,a_{k+1}$, we find that the slope has to
  increase at every step.
\end{sketch}

\newcommand{\lowerStraightThreeRegText}{%
	There exist 3-regular 3-connected 1-plane graphs such that any 
	embedding-preserving straight-line drawing uses at least $18$ distinct slopes.}
\begin{lemma}\label{lem:lower-straight-3reg}
  \lowerStraightThreeRegText
\end{lemma}
\begin{sketch}
	Consider the graph depicted in Fig.~\ref{fig:lowerbound-3reg-main}. 
  We find that the slopes of the edges $(a_i,b_i),(a_i,c_i),(c_i,d_i),(c_i,e_i),(e_i,d_i),(e_i,a_{i+1})$
  have to be increasing in this order for every~$i=1,2,3$.
\end{sketch}

\newcommand{\lowerStraightThreeGenText}{%
	There exist 3-connected 1-plane graphs such that 
	any embedding-preserving straight-line drawing uses at least $9(\Delta-1)$ 
	distinct slopes.}
\begin{theorem}\label{thm:lower-straight-gen}
	\lowerStraightThreeGenText
\end{theorem}
\begin{sketch}
 Consider the graph depicted in Fig.~\ref{fig:lowerbound-gen-main}. 
 The degree of $a_i$, $c_i$, and~$e_i$ is $\Delta$.
 We repeat the proof of Lemma~\ref{lem:lower-straight-3reg}, but observe that
 the slopes of the $9(\Delta-3)$ added edges lie between the slopes of $(a_i,b_i),(a_i,c_i),(c_i,e_i)$, and $(e_i,a_{i+1})$.
 \end{sketch}


\section{Open problems}\label{sec:open}

The research in this paper gives rise to interesting questions, among them: (1) Is it possible to extend Theorem~\ref{thm:3con1bend} to all subcubic $1$-planar graphs? (2) Can we drop the embedding-preserving condition from Theorem~\ref{thm:lower-1bend}? (3) Is the $1$-planar slope number of $1$-planar graphs bounded by a function of the maximum degree?


\ifproc
\clearpage
\fi

\bibliographystyle{splncs04}
\bibliography{abbrv,1planarslopes}

\begin{thebibliography}{10}
\providecommand{\url}[1]{\texttt{#1}}
\providecommand{\urlprefix}{URL }
\providecommand{\doi}[1]{https://doi.org/#1}

\bibitem{DBLP:conf/gd/AlamBK13}
Alam, M.J., Brandenburg, F.J., Kobourov, S.G.: Straight-line grid drawings of
  3-connected 1-planar graphs. In: Wismath, S.K., Wolff, A. (eds.) Proc. 21st
  Int. Symp. Graph Drawing (GD'13). Lecture Notes Comput. Sci., vol.~8242, pp.
  83--94. Springer (2013). \doi{10.1007/978-3-319-03841-4\_8}

\bibitem{DBLP:conf/compgeom/AngeliniBLM17}
Angelini, P., Bekos, M.A., Liotta, G., Montecchiani, F.: A universal slope set
  for 1-bend planar drawings. In: Aronov, B., Katz, M.J. (eds.) Proc. 33rd Int.
  Symp. Comput. Geom. (SoCG'17). LIPIcs, vol.~77, pp. 9:1--9:16. Schloss
  Dagstuhl (2017). \doi{10.4230/LIPIcs.SoCG.2017.9}

\bibitem{DBLP:journals/combinatorics/BaratMW06}
Bar{\'a}t, J., Matousek, J., Wood, D.R.: Bounded-degree graphs have arbitrarily
  large geometric thickness. Electr. J. Comb.  \textbf{13}(1),  1--14 (2006),
  \url{http://www.combinatorics.org/Volume\_13/Abstracts/v13i1r3.html}

\bibitem{DBLP:journals/tcs/BekosDLMM17}
Bekos, M.A., Didimo, W., Liotta, G., Mehrabi, S., Montecchiani, F.: On {RAC}
  drawings of 1-planar graphs. Theor. Comput. Sci.  \textbf{689},  48--57
  (2017). \doi{10.1016/j.tcs.2017.05.039}

\bibitem{DBLP:journals/jgaa/BekosG0015}
Bekos, M.A., Gronemann, M., Kaufmann, M., Krug, R.: Planar octilinear drawings
  with one bend per edge. J. Graph Algorithms Appl.  \textbf{19}(2),  657--680
  (2015). \doi{10.7155/jgaa.00369}

\bibitem{bk-bhogd-cgta98}
Biedl, T., Kant, G.: A better heuristic for orthogonal graph drawings. Comput.
  Geom. Theory Appl.  \textbf{9}(3),  159--180 (1998).
  \doi{10.1016/S0925-7721(97)00026-6}

\bibitem{DBLP:journals/comgeo/Brandenburg18}
Brandenburg, F.J.: {T}-shape visibility representations of 1-planar graphs.
  Comput. Geom.  \textbf{69},  16--30 (2018).
  \doi{10.1016/j.comgeo.2017.10.007}

\bibitem{DBLP:conf/ewcg/Chaplick18}
Chaplick, S., Lipp, F., Wolff, A., Zink, J.: 1-bend {RAC} drawings of
  {N}{I}{C}-planar graphs in quadratic area. In: Korman, M., Mulzer, W. (eds.)
  Proc. 34th Europ. Workshop Comput. Geom. (EuroCG'18). pp. 28:1--28:6. FU
  Berlin (2018),
  \url{https://conference.imp.fu-berlin.de/eurocg18/download/paper_28.pdf}

\bibitem{DBLP:books/daglib/0023376}
Cormen, T.H., Leiserson, C.E., Rivest, R.L., Stein, C.: Introduction to
  Algorithms (3rd ed.). {MIT} Press (2009),
  \url{https://mitpress.mit.edu/books/introduction-algorithms-third-edition}

\bibitem{DBLP:journals/tcs/BattistaT88}
{Di Battista}, G., Tamassia, R.: Algorithms for plane representations of
  acyclic digraphs. Theor. Comput. Sci.  \textbf{61},  175--198 (1988).
  \doi{10.1016/0304-3975(88)90123-5}

\bibitem{DBLP:journals/algorithmica/GiacomoDELMMW18}
{Di Giacomo}, E., Didimo, W., Evans, W.S., Liotta, G., Meijer, H.,
  Montecchiani, F., Wismath, S.K.: Ortho-polygon visibility representations of
  embedded graphs. Algorithmica  \textbf{80}(8),  2345--2383 (2018).
  \doi{10.1007/s00453-017-0324-2}

\bibitem{DBLP:journals/jgaa/GiacomoLM15}
{Di Giacomo}, E., Liotta, G., Montecchiani, F.: Drawing outer 1-planar graphs
  with few slopes. J. Graph Algorithms Appl.  \textbf{19}(2),  707--741 (2015).
  \doi{10.7155/jgaa.00376}

\bibitem{DBLP:journals/tcs/GiacomoLM18}
{Di Giacomo}, E., Liotta, G., Montecchiani, F.: Drawing subcubic planar graphs
  with four slopes and optimal angular resolution. Theor. Comput. Sci.
  \textbf{714},  51--73 (2018). \doi{10.1016/j.tcs.2017.12.004}

\bibitem{DBLP:journals/corr/abs-1804-07257}
Didimo, W., Liotta, G., Montecchiani, F.: A survey on graph drawing beyond
  planarity. Arxiv report 1804.07257 (2018),
  \url{https://arxiv.org/abs/1804.07257}

\bibitem{DBLP:journals/comgeo/DujmovicESW07}
Dujmovi\'{c}, V., Eppstein, D., Suderman, M., Wood, D.R.: Drawings of planar
  graphs with few slopes and segments. Comput. Geom.  \textbf{38}(3),  194--212
  (2007). \doi{10.1016/j.comgeo.2006.09.002}

\bibitem{orthoChapterHandbook}
Duncan, C., Goodrich, M.T.: Planar orthogonal and polyline drawing algorithms.
  In: Tamassia, R. (ed.) Handbook on Graph Drawing and Visualization. Chapman
  and Hall/CRC (2013),
  \url{http://cs.brown.edu/people/rtamassi/gdhandbook/chapters/orthogonal.pdf}

\bibitem{DBLP:journals/jgaa/DuncanK03}
Duncan, C.A., Kobourov, S.G.: Polar coordinate drawing of planar graphs with
  good angular resolution. J. Graph Algorithms Appl.  \textbf{7}(4),  311--333
  (2003). \doi{10.7155/jgaa.00073}

\bibitem{DBLP:journals/dam/EadesL13}
Eades, P., Liotta, G.: Right angle crossing graphs and $1$-planarity. Discrete
  Appl. Math.  \textbf{161}(7-8),  961--969 (2013).
  \doi{10.1016/j.dam.2012.11.019}

\bibitem{DBLP:journals/dm/FabriciM07}
Fabrici, I., Madaras, T.: The structure of 1-planar graphs. Discrete Math.
  \textbf{307}(7-8),  854--865 (2007). \doi{10.1016/j.disc.2005.11.056}

\bibitem{DBLP:journals/siamcomp/FormannHHKLSWW93}
Formann, M., Hagerup, T., Haralambides, J., Kaufmann, M., Leighton, F.T.,
  Symvonis, A., Welzl, E., Woeginger, G.J.: Drawing graphs in the plane with
  high resolution. SIAM J. Comput.  \textbf{22}(5),  1035--1052 (1993).
  \doi{10.1137/0222063}

\bibitem{DBLP:conf/cocoon/HongELP12}
Hong, S., Eades, P., Liotta, G., Poon, S.: F{\'{a}}ry's theorem for 1-planar
  graphs. In: Gudmundsson, J., Mestre, J., Viglas, T. (eds.) Proc. 18th Ann.
  Int. Conf. Comput. Comb. (COCOON'12). Lecture Notes in Computer Science,
  vol.~7434, pp. 335--346. Springer (2012). \doi{10.1007/978-3-642-32241-9\_29}

\bibitem{DBLP:journals/gc/JelinekJKLTV13}
Jelinek, V., Jelinkov{\'a}, E., Kratochvil, J., Lidick{\'y}, B., Tesar, M.,
  Vyskocil, T.: The planar slope number of planar partial 3-trees of bounded
  degree. Graphs Comb.  \textbf{29}(4),  981--1005 (2013).
  \doi{10.1007/s00373-012-1157-z}

\bibitem{k-hgd-wg91}
Kant, G.: Hexagonal grid drawings. In: Mayr, E.W. (ed.) Proc. 18th Int.
  Workshop Graph-Theor. Concepts Comput. Sci. (WG'92). Lecture Notes Comput.
  Sci., vol.~657, pp. 263--276. Springer (1992).
  \doi{10.1007/3-540-56402-0\_53}

\bibitem{DBLP:journals/algorithmica/Kant96}
Kant, G.: Drawing planar graphs using the canonical ordering. Algorithmica
  \textbf{16}(1),  4--32 (1996). \doi{10.1007/BF02086606}

\bibitem{DBLP:journals/siamdm/KeszeghPP13}
Keszegh, B., Pach, J., P{\'a}lv{\"o}lgyi, D.: Drawing planar graphs of bounded
  degree with few slopes. SIAM J. Discrete Math.  \textbf{27}(2),  1171--1183
  (2013). \doi{10.1137/100815001}

\bibitem{DBLP:journals/comgeo/KnauerMW14}
Knauer, K.B., Micek, P., Walczak, B.: Outerplanar graph drawings with few
  slopes. Comput. Geom.  \textbf{47}(5),  614--624 (2014).
  \doi{10.1016/j.comgeo.2014.01.003}

\bibitem{DBLP:journals/csr/KobourovLM17}
Kobourov, S.G., Liotta, G., Montecchiani, F.: An annotated bibliography on
  1-planarity. Comput. Sci. Reviews  \textbf{25},  49--67 (2017).
  \doi{10.1016/j.cosrev.2017.06.002}

\bibitem{DBLP:conf/gd/LenhartLMN13}
Lenhart, W., Liotta, G., Mondal, D., Nishat, R.I.: Planar and plane slope
  number of partial 2-trees. In: Wismath, S.K., Wolff, A. (eds.) Proc. 21st
  Int. Symp. Graph Drawing (GD'13). Lecture Notes Comput. Sci., vol.~8242, pp.
  412--423. Springer (2013). \doi{10.1007/978-3-319-03841-4\_36}

\bibitem{lms-la2be-DAM98}
Liu, Y., Morgana, A., Simeone, B.: A linear algorithm for 2-bend embeddings of
  planar graphs in the two-dimensional grid. Discrete Appl. Math.
  \textbf{81}(1--3),  69--91 (1998). \doi{10.1016/S0166-218X(97)00076-0}

\bibitem{DBLP:journals/siamdm/MalitzP94}
Malitz, S.M., Papakostas, A.: On the angular resolution of planar graphs. SIAM
  J. Discrete Math.  \textbf{7}(2),  172--183 (1994).
  \doi{10.1137/S0895480193242931}

\bibitem{DBLP:conf/gd/MukkamalaP11}
Mukkamala, P., P{\'{a}}lv{\"{o}}lgyi, D.: Drawing cubic graphs with the four
  basic slopes. In: van Kreveld, M.J., Speckmann, B. (eds.) Proc. 19th Int.
  Symp. Graph Drawing (GD'11). Lecture Notes Comput. Sci., vol.~7034, pp.
  254--265. Springer (2011). \doi{10.1007/978-3-642-25878-7\_25}

\bibitem{DBLP:journals/combinatorics/PachP06}
Pach, J., P{\'a}lv{\"o}lgyi, D.: Bounded-degree graphs can have arbitrarily
  large slope numbers. Electr. J. Comb.  \textbf{13}(1), ~1--4 (2006),
  \url{http://www.combinatorics.org/Volume\_13/Abstracts/v13i1n1.html}

\bibitem{R65}
Ringel, G.: Ein {S}echsfarbenproblem auf der {K}ugel. Abh. aus dem Math.
  Seminar der Univ. Hamburg  \textbf{29}(1--2),  107--117 (1965).
  \doi{10.1007/BF02996313}

\end{thebibliography}

\ifproc
\end{document}
\fi


\clearpage
\appendix

\section{Omitted material from Section~\ref{sec:prelim}}\label{app:prelim}

\begin{backInTime}{crossingMinimal}
\begin{lemma}
  \crossingMinimalText
\end{lemma}
\begin{proof}
  We first show how to iteratively remove all cutvertices from~$G^*$ that are
  dummy vertices. 
  Suppose that there is a cutvertex~$v$ in~$G^*$ that is a dummy vertex.
  Let~$a,b,c,d$ be the neighbors of~$v$ in counter-clockwise order, so
  the edges $(a,c)$ and $(b,d)$ cross in~$G$.
  
  \begin{figure}[b]
    \centering
    \subcaptionbox{$(v,a)$ is a bridge\label{fig:lemma1-bridge}}{\includegraphics[page=1]{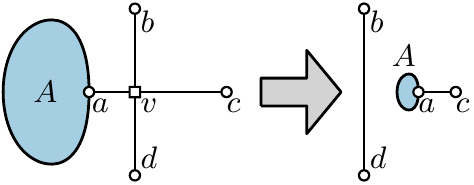}}
    \hfill
    \subcaptionbox{$v$ has no bridge\label{fig:lemma1-cutvertex}}{\includegraphics[page=2]{lemma1}}
    \caption{Eliminating dummy cutvertices from~$G^*$}
    \label{fig:lemma1}
  \end{figure}
  
  First, assume that one of the four edges of~$v$, say $(v,a)$, is a bridge,
  so removing~$v$ from~$G^*$ gives a connected component~$A$ that contains
  neither~$b,c,d$; see Fig.~\ref{fig:lemma1-bridge}.  We shrink~$A$
  and move it to the other side of~$v$. This eliminates the crossing between $(a,c)$
  and $(b,d)$ from~$G$. 
  
  We can now assume that there is no bridge at~$v$, so removal of~$v$ divides~$G^*$
  into two components~$A$ and~$B$. If~$a$ and~$c$ lie in the same component,
  $G$ is disconnected as there is no path from~$a$ to~$b$; hence, assume
  w.l.o.g. that~$a,b\in A$ and $c,d\in B$; see Fig.~\ref{fig:lemma1-cutvertex}.
  We flip~$B$, reroute the edge~$(a,c)$ along $(a,v)$ and $(v,c)$ and reroute
  $(b,d)$ along $(b,v)$ and $(v,d)$. This eliminates the crossing between
  $(a,c)$ and $(b,d)$ from~$G$.
  
  This shows the first part of the lemma. For the second part, suppose that~$G$
  is 3-connected and~$G^*$ has no dummy vertex as a cutvertex; otherwise, 
  apply the first part of the lemma. Assume that there is a separation
  pair~$u,v$ in~$G^*$ where~$v$ is a dummy vertex. Let again~$a,b,c,d$ be the 
  neighbors of~$v$ in counter-clockwise order. 
  
  First, assume that~$u$ is a real vertex. 
  Removal of~$u$ and~$v$ splits~$G^*$ in at most four connected
  components. If one of these connected components contains exactly one 
  neighbor of~$v$, say~$a$, there are at most two vertex-connected paths
  from~$a$ to~$c$ in~$G$: the edge $(a,c)$ and one path via~$u$. Hence, 
  there are two connected components~$A$ and~$B$ that contain two neighbors of~$v$
  each; assume w.l.o.g. that~$A$ contains~$a$. If~$A$ contains~$a$ and~$c$, 
  $u$ is a cutvertex in~$G$, which contradicts 3-connectivity.
  If~$A$ contains~$a$ and~$b$, we flip~$A$, reroute the 
  edge~$(b,d)$ along $(b,v)$ and $(v,d)$ and reroute
  $(a,c)$ along $(a,v)$ and $(v,c)$; see Fig.~\ref{fig:lemma1dummy-dummyrealpair}. This eliminates the crossing between
  $(a,c)$ and $(b,d)$ from~$G$. If~$A$ contains~$a$ and~$d$, we proceed
  analogously.
  
  Second, assume that~$u$ is also a dummy vertex with neighbors $a',b',c',d'$ 
  in counter-clockwise order. Removal of~$u$ and~$v$ splits~$G^*$ in at most four connected
  components. If one of these connected components contains exactly one 
  neighbor of~$v$, say~$a$, and exactly one neighbor of~$u$, say~$a'$, then 
  there are at most two vertex-connected paths
  from~$a$ to~$c$ in~$G$: the edge $(a,c)$ and one path via the edge~$(a',c')$. 
  If one of the connected components contains exactly three neighbors of~$v$,
  say~$a,b,c$, and exactly one neighbor of~$u$, say~$a'$, then $d,a'$ is
  a separation pair in~$G$, as it separates $a,b,c$ from $b',c',d'$.
  Hence, each of these connected components contains exactly two neighbors
  of one of~$v$ and~$u$. Let~$A$ be a connected component and assume w.l.o.g.
  that it contains~$a$ and one more neighbor of~$v$.
  If~$A$ contains~$a$ and~$c$, there is some neighbor of~$u$ that is not
  in~$A$, say~$a'$. Since~$(a,c)\in A$ and $b,d\notin A$, all paths from~$a$
  to~$a'$ in~$G$ have to traverse the edge $(a',c')$ or $(b',d')$, so there
  are at most two of them; a contradiction to 3-connectivity.
  If~$A$ contains~$a$ and~$b$, we flip~$A$ and reroute the 
  edges~$(a,c)$ and $(b,d)$ to eliminate their crossing from~$G$. Note that,
  if~$A$ contains also exactly two neighbors of~$u$,  this also
  eliminates the crossing between $(a',c')$ and $(b',d')$; see Fig.~\ref{fig:lemma1dummy-dummydummypair}. If~$A$ contains~$a$ and~$d$, we proceed
  analogously.

	\begin{figure}[t]
    \centering
		\subcaptionbox{$A$ contains $a,b$ and $u$ is real vertex.\label{fig:lemma1dummy-dummyrealpair}}{\includegraphics[page=1]{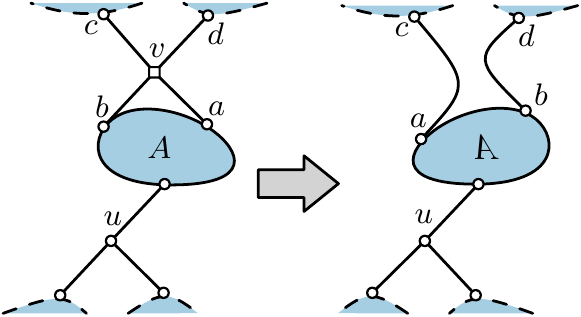}}
				 \hfill
    \subcaptionbox{$A$ contains $a,b$ and $u$ is dummy vertex.\label{fig:lemma1dummy-dummydummypair}}{~~\includegraphics[page=2]{lemma1Dummy}~~}
    \caption{Eliminating a dummy separation pair from~$G^*$}
    \label{}
  \end{figure}
  
  Each step reduces the number of crossings in the embedding, so it terminates 
  with an embedding that has a 3-connected planarization.
\end{proof}
\end{backInTime}

\section{Omitted material from Section~\ref{sec:3con1bend}}\label{app:3con1bend}

\begin{backInTime}{stretch}
\begin{figure}[t]
  \begin{subfigure}[b]{.47\textwidth}
    \centering
    \includegraphics[page=5]{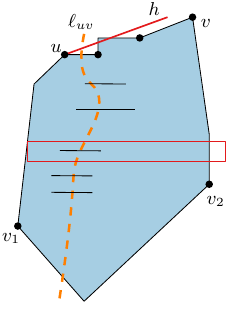}
    \caption{}
    \label{fig:stretch-1}
  \end{subfigure}
  \hfil
  \begin{subfigure}[b]{.47\textwidth}
    \centering
    \includegraphics[page=6]{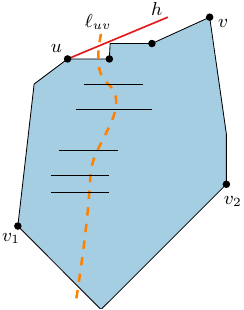}
    \caption{}
    \label{fig:stretch-2}
  \end{subfigure}
  \caption{Illustration for Lemma~\ref{le:stretch}}
  \label{fig:stretch}
\end{figure}

\begin{lemma}
  \stretchText
\end{lemma}
\begin{proof}
Refer to Fig.~\ref{fig:stretch}.  Suppose there is a half-line $h$ that 
originates at $u$ (the argument is analogous for $v$)  with  slope in the set 
$\{\pi/4, \pi/2\}$ that intersects the outer face of $\Gamma_i$.
Let~$s$ be the first edge segment of~$P_i(u,v)$ that is intersected by~$h$.
By the slopes of~$h$, we have that~$s$ is drawn with slope~$\pi/2$. Then,~$s$
must be traversed upwards when going from $u$ to $v$ in $P_i(u,v)$, and thus 
by \ref{P4} there is a horizontal segment before~$s$.

Let~$e$ be the edge containing this horizontal segment. By \ref{P4}, there is a $uv$-cut 
with respect to $e$ and there is a $y$-monotone curve $\ell_{uv}$ 
that cuts the horizontal segments of this cut. Let $C_u$ and $C_v$ be the two 
components defined by the $uv$-cut, such that $C_u$ contains $u$ and $C_v$ 
contains $v$. We shift all vertices in $C_v$ and all edges having both end-vertices 
in $C_v$ to the right by $\sigma$ units, for some suitable $\sigma>0$. 
All vertices in $C_u$ and all edges having both end-vertices in $C_u$ are 
not modified. Furthermore, the edges having an end-vertex in $C_u$ and the 
other end-vertex in $C_v$ are all and only the edges of the $uv$-cut, and 
thus they all contain a horizontal segment in $\Gamma_i$ that can be 
stretched by $\sigma$ units. Finally, note that $(v_1,v_2)$ is also part of 
the $uv$-cut, but it does not contain any horizontal segment; however, by \ref{P3} 
its two segments can be always redrawn by using the SE port of~$v_1$ and the 
SW port of~$v_2$. For a suitable choice of $\sigma$, this operation removes 
the crossing between $h$ and~$s$.
Moreover, no new edge crossing can appear in the drawing because $\ell_{uv}$ 
intersects only the edge segments of the cut. 
Hence, we can repeat this procedure until all crossings between $h$ and 
segments of~$P_i(u,v)$ are resolved. 
The resulting drawing is clearly still valid and stretchable.
\end{proof}
\end{backInTime}

\begin{backInTime}{gkvalid}
\begin{lemma}
  \gkvalidText
\end{lemma}
\begin{proof}
In all the cases used by our construction, we guaranteed the drawing 
be valid and attachable. Concerning stretchability, observe that \ref{P3} is guaranteed 
by Lemma~\ref{le:stretch}. To show \ref{P4}, one can use induction on  $i \le K-1$ as 
follows. 

In the base case $i=2$, we have that $\Gamma_2$ is clearly 
stretchable by construction. When adding $\V_i$ to $\Gamma_{i-1}$, 
we have that \ref{P4} holds by induction for all pairs of vertices that are consecutive 
both in $\Gamma_{i-1}$ and in $\Gamma_i$, because $P_{i-1}(u,v)=P_i(u,v)$. 
Also, the vertices in $\V_i$ are all attachable vertices. We distinguish the following
cases.

\setcounter{casecounter}{0}
\ccase{c:stretch-singleton} $\V_i$ is singleton. Then,~$v^i$ is consecutive with either $u_l$ or the attachable 
vertex $w$ before~$u_l$, and with either $u_r$ or the attachable vertex $w'$ after $u_r$. 

\subcase{sc:stretch-signleton-ul} $v^i$ is consecutive to $u_l$ (resp., $u_r$). Then, $u_l$ (resp., $u_r$) has degree four and hence is dummy. 
However, now~$u_l$ (resp., $u_r$) is not L-attachable (resp., R-attachable) anymore,
so \ref{P4} holds.

\subcase{sc:stretch-signleton-w} $v^i$ is consecutive to $w$ (a symmetric argument applies to $w'$). Then, observe 
first that \ref{P4} holds for $P_{i-1}(w,u_l) = P_i(w,u_l)$. 
Also, observe that if there is a horizontal segment in $P_i(w,u_l)$, then \ref{P4} holds for $P_i(w,v^i)$ (even if $(u_l,v^i)$ contains a vertical segment and even if $v^i$ is real). 

\subsubcase{ssc:stretch-signleton-w-rr} Both $w$ and $u_l$ are real. Then, there is a 
  horizontal segment in $P_i(w,u_l)$ and hence \ref{P4} holds for $P_i(w,v^i)$. 

\subsubcase{ssc:stretch-signleton-w-dr} $w$ is dummy and~$u_l$ is real. Then,~$u_l$ is R-attachable. 
  Furthermore, this means that $(u_l,v^i)$ will not be drawn as a vertical segment. 
  Hence, \ref{P4} holds for $P_i(w,v^i)$.
  
\subsubcase{ssc:stretch-signleton-w-dd} Both~$w$ and~$u_l$ are dummy. Then,~$v^i$ is real. 
  If~$v^i$ is the only successor
  of~$u_l$, then it is drawn with Case C3L in Fig.~\ref{fig:Cx}, so there is a
  horizontal segment on $(u_l,v^i)$. 
  Otherwise,~$u_l$ has another successor~$x$ in~$G_i$. If~$x$ is an R-successor
  of~$u_l$, then~$v^i$ is drawn with Case C1R, C2R, C2sR, or C3sR in Fig.~\ref{fig:Cxx}, column~2.
  In each case, $(u_l,v^i)$ is drawn with a horizontal segment, so \ref{P4} holds.
  If~$x$ is an L-successor of~$u_l$, then~$v^i$ is drawn with Case C1L, C2L, C2sL,
  or C3L in Fig.~\ref{fig:Cxx}, column~2. In each case, the edge $(u_l,v^i)$ goes
  downward from~$v^i$.

\begin{figure}[t]
  \centering
	\includegraphics{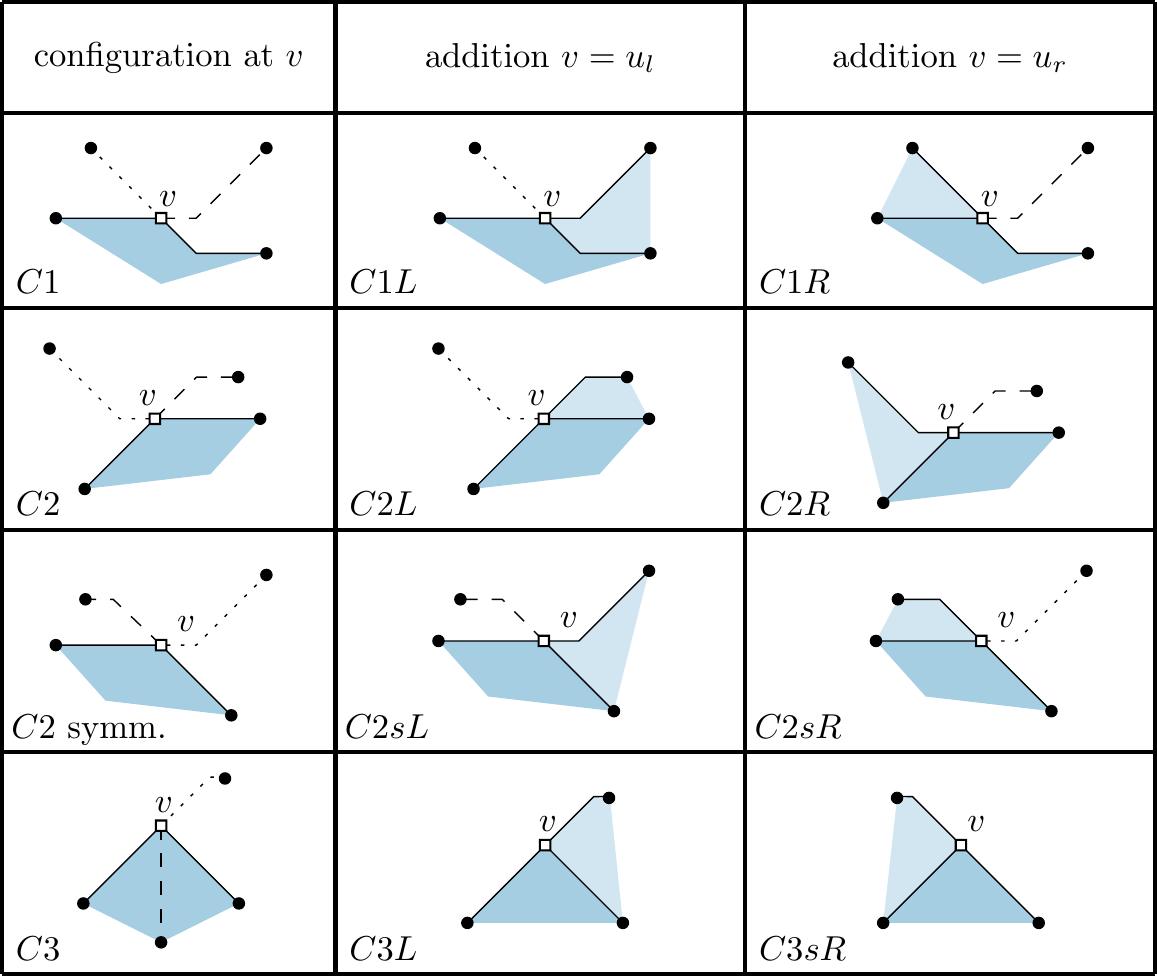}
	\caption{Cases for dummy vertices with one successor}
	\label{fig:Cx}
\end{figure}
  
  We move along $P(w,v^i)$ from~$v^i$ to~$w$ until we
  either arrive at an edge of a chain, at~$w$, or at a dummy vertex~$y\neq w$. 
  Since a real vertex cannot
  have two successors, this walk only moves downwards. 
  Recall that all edges of a chain are drawn with a horizontal segment, 
  so if we arrive at such an edge, then \ref{P4} holds. 
  
  If we arrive at
  a dummy vertex~$y\neq w$, then it has an L-successor in~$G_i$. 
  Since~$y$ and~$v^i$ are not consecutive,
  all of its successors are already drawn, so we are in one of the cases
  of Fig.~\ref{fig:Cxx}, column~2 or~3, that has at least one L-successor.
  Then, either one of the edges of~$y$ on $C_i$ has a horizontal segment,
  so \ref{P4} holds, or the edge from~$y$ goes downwards (Cases C2L and C3L in column~2), 
  and we can continue our walk. 
  
  If we arrive at~$w$ then the last edge we traverse was 
  to~$w$ from an L-successor of~$w$ in~$G_i$. Hence,~$w$ is not L-attachable.
  Further, since our walk only went downwards from~$v^i$ to~$w$, there was no upwards 
  vertical on this segment. Hence, \ref{P4} holds.

\subsubcase{ssc:stretch-signleton-w-rd} $w$ is real and~$u_l$ is dummy, then~$v^i$ is real.
  We proceed exactly as in Case~\ref{ssc:stretch-signleton-w-dd}. However, since~$w$
  is real and attachable, the last edge cannot be from a successor of~$w$ to~$w$.
  Hence, it is either an edge of a chain, so it has a horizontal segment, or on
  the way we encountered a dummy vertex that has an edge with a horizontal segment.
  In either case, \ref{P4} holds.

\ccase{c:stretch-chain} $\V_i$ is a chain.
  Note first that \ref{P4} holds for each pair of consecutive vertices $u$ and $v$ such 
  that both of them are in $\V_i$, since all edges of $\V_i$ contain a horizontal 
  segment. If $u \in \V_j$ and $v \in \V_i$ (resp., $u \in \V_i$ and $v \in \V_j$) 
  with $j<i$, then $v=v^i_1$ (resp., $u=v^i_l$), and a similar argument as for the 
  singletons can be applied. 
  This completes the case analysis.

\begin{figure}
  \centering
	\includegraphics{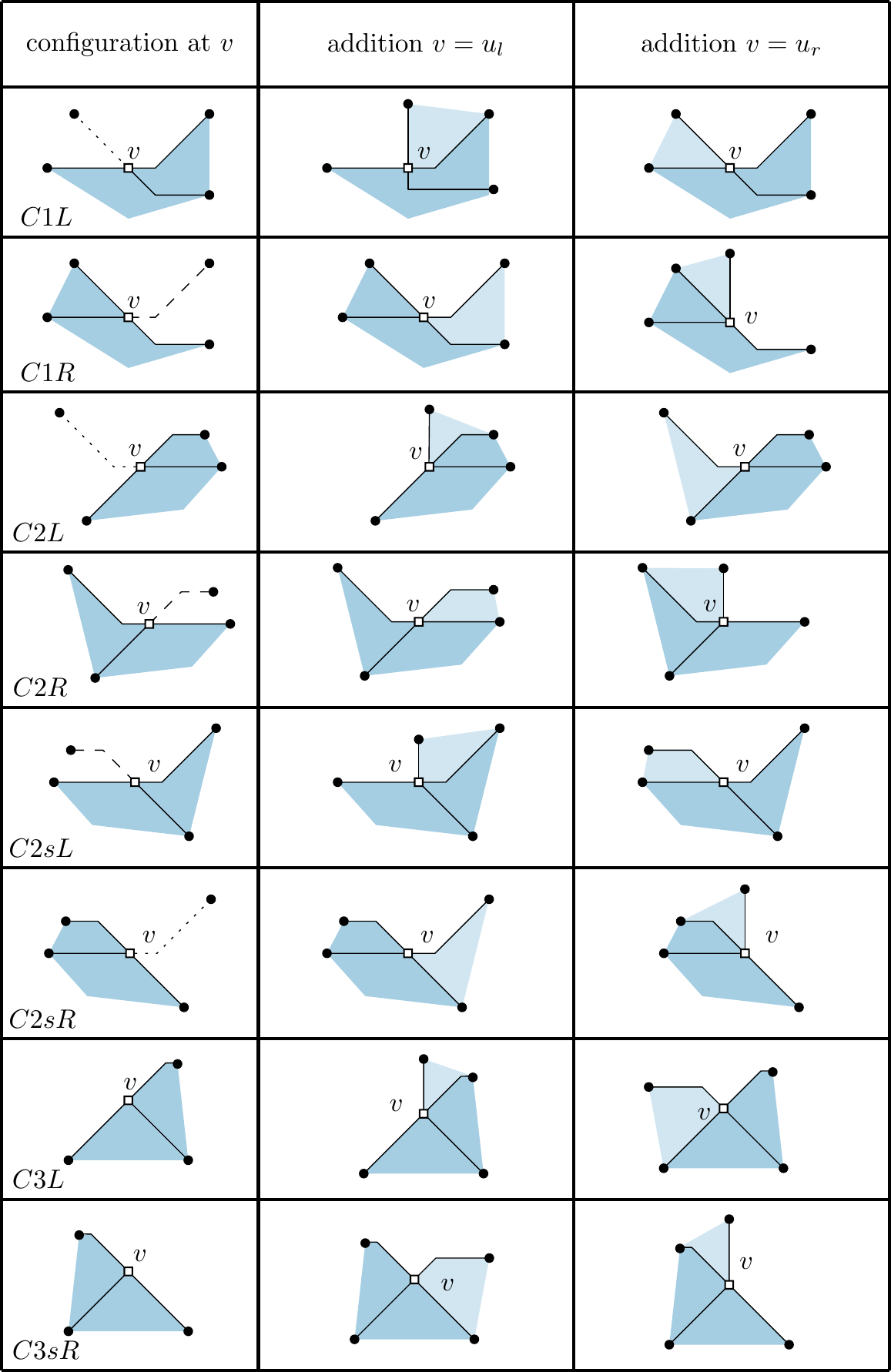}
	\caption{Cases for dummy vertices with two successors}
	\label{fig:Cxx}
\end{figure}

Finally, we should prove that for every edge $e$ of $P_i(u,v)$ such that $e$ contains a 
horizontal segment, there exists a $uv$-cut of $G_i$ with respect to $e$ 
whose edges all contain a horizontal segment in $\Gamma_i$ except 
for $(v_1,v_2)$, and such that there exists a $y$-monotone curve that passes through all 
and only such horizontal segments and through $(v_1,v_2)$. Again, 
this is true by induction for all such pairs of vertices that are consecutive  both in $\Gamma_{i-1}$ and in $\Gamma_i$, 
because $P_{i-1}(u,v)=P_i(u,v)$. 

If both $u \in \V_i$ and $v \in \V_i$, then $\V_i$ is a chain. 
 Let $u_l$ and $u_r$ be the two vertices on $C_{i-1}$ used by $\V_i$ to attach 
to $\Gamma_{i-1}$. Note that, by \ref{P4}, $P_{i-1}(u_l,u_r)$ does not contain any horizontal segment only 
if it contains neither pairs of consecutive real vertices, nor vertical segments. 
It is not difficult to see that this situation never occurs, and hence there is
at least an edge $e$ in $P_{i-1}(u_l,u_r)$ such that $e$ 
contains a horizontal segment, there exists a $uv$-cut of~$G_{i-1}$ with 
respect to $e$ whose edges all contain a horizontal segment in $\Gamma_{i-1}$ 
except for $(v_1,v_2)$, and such that there exists a $y$-monotone curve that 
passes through all and only such horizontal segments and through $(v_1,v_2)$. 
Now, we can add the edge $(u,v)$ to this cut, since this edge contains a 
horizontal segment by construction. Also, $(u,v)$ is on the same face 
of $\Gamma_i$ as $e$ and above it, so the $y$-monotone curve that passes 
through $e$ can be suitably modified so to also pass through the horizontal 
segment of $(u,v)$. 

If $u \in \Gamma_{i-1}$ and $v \in \V_i$ (the symmetric case is 
analogous), as explained above, we have that $P_i(u,v)$ is constructed 
from $P_{i-1}(u,u_l)$, where $u_l$ is the leftmost predecessor of $v$ (which is 
either a singleton $v^i$ or the first vertex $v^i_1$ of a chain). Then, the 
only edge for which we may need to prove the property is the edge $(u_l,v)$ 
and only if it contains a horizontal segment. If so, we can again exploit the 
fact that there is at least an edge $e$ in $P_{i-1}(u_l,u_r)$ (where $u_r$ is 
the rightmost predecessor of $v$ if $\V_i$ is a singleton or of $v^i_l$ 
if $\V_i$ is a chain) for which \ref{P4} holds by induction.
\end{proof}
\end{backInTime}

\section{Omitted proofs from Section~\ref{sec:2bends}}\label{app:2bends}

\begin{backInTime}{bendInvariants}
\begin{lemma}
  \bendInvariantsText
\end{lemma}
\begin{proof}
  By construction, $G_i^*$ is biconnected. 
  First, observe that every face in~$G_i^*$
  contains at least two real vertices since no two dummy vertices can be 
  adjacent. Hence, there is some face that contains the real vertex~$t=u_i$ and some real vertex~$s$.
  We use these two vertices to compute an $st$-order~$\sigma$ and use the algorithm of
  Liu et al.\ to draw~$G_i^*$. We first show that this drawing satisfies all
  invariants. Invariants (I1) and (I2) are trivially satisfied.
  
  Since~$s$ and~$t$ are real vertices, both have degree at most~3, so there
  are no U-shapes in the drawing; since U-shapes are the only edges that are
  not drawn $y$-monotone from its source to its target, this satisfies (I3). 
  By construction, all edges in~$G_i^*$ are
  drawn with at most 2 bends; hence, invariant (I4) holds.
  
  Consider a dummy vertex~$v$ in~$G_i^*$ with neighbors $a,b,c,d$ in
  clockwise order; hence, the edges $(a,c)$ and $(b,d)$ cross in the
  given embedding of~$G_i$. Assume w.l.o.g. that~$(v,a)$ uses the S port,
  $(v,b)$ uses the W port, $(v,c)$ uses the N port, and $(v,d)$ uses
  the E port at~$v$. Since there are no U-shapes in the drawing,
  both $(v,a)$ and $(v,c)$ have to be drawn as a vertical or an L-shape,
  so they both have at most 1 bend.
  
  Consider now the edges $(v,b)$ and $(v,d)$. Both edges are drawn as a horizontal,
  an L-shape, or a C-shape. Recall that~$b$ and~$c$ are real vertices, so they
  have at most degree~3. If~$(v,b)$ starts in~$b$, it uses the N or the E port
  at~$b$, but it uses the W port at~$v$, so it cannot be a C-shape.
  If $(v,b)$ ends in~$b$, it uses the W or the S port at~$b$,
  and it can only be a C-shape if it uses the W one.
  Symmetrically, If~$(v,d)$ starts in~$d$, it uses the N port or the E port
  at~$d$ and it can only be a C-shape if it uses the E one.
  If $(v,d)$ ends in~$d$, it uses the W or the S port at~$d$,
  but it uses the E port at~$v$, so it cannot be a C-shape.
  This establishes invariants (I5) and (I6) and proofs the lemma.
\end{proof}
\end{backInTime}

\begin{backInTime}{bendBicon}
\begin{lemma}
  \bendBiconText
\end{lemma}
\begin{proof}
  It remains to show the cases that the N port at~$u$ is not free.
  \setcounter{casecounter}{1}
  
  \ccase{c:bottomtop} The N port at~$u$ is used by an edge $(u,w)$ and the W port is free.
  We distinguish three more cases based on the drawing style of $(u,w)$.
  
  \begin{figure}[t]
    \centering
    \subcaptionbox{}{\includegraphics[page=4]{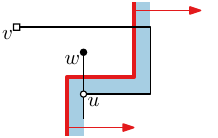}}
    \hfill
    \subcaptionbox{}{\includegraphics[page=5]{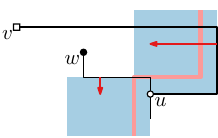}}
    \hfill
    \subcaptionbox{}{\includegraphics[page=6]{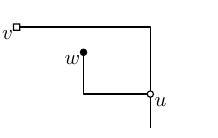}}
    \caption{Proof of Lemma~\ref{lem:2bend-bicon}, Case~\ref{sc:bottomtop-up}}
    \label{fig:2bend-bottomtop-up}
  \end{figure}
  
  \subcase{sc:bottomtop-up} $(u,w)$ is a vertical edge;
  see Fig.~\ref{fig:2bend-bottomtop-up}.
  We create a curve~$S$ as in Case~\ref{c:bottomleft} except that we do not
  pass the vertical segment of~$(u,v)$ but extend it upwards to infinity before.
  We stretch the drawing by moving~$S$ to the right such that~$u$ is
  placed below the top-right bend point of $(u,v)$. Now the edge
  $(u,w)$ is drawn with 2 bends, but the area between~$u$ and the two bend
  points is empty and the W port of~$u$ is unused, so we can make an L-shape out
  of $(u,w)$ that uses the W port at~$u$. Furthermore,
  similar to Case~\ref{c:bottomleft}, the region between~$u$ and the top-right
  bend point of $(u,v)$ is free and now the N port of~$u$ is unused, 
  so we can make an L-shape out of $(u,v)$ that uses the N port at~$u$. 
  
  \begin{figure}[b]
    \centering
    \subcaptionbox{}{\includegraphics[page=7]{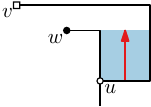}}
    \hfill
    \subcaptionbox{}{\includegraphics[page=8]{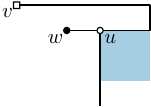}}
    \hfill
    \subcaptionbox{}{\includegraphics[page=9]{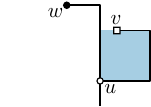}}
    \caption{Proof of Lemma~\ref{lem:2bend-bicon}, Case~\ref{sc:bottomtop-left}}
    \label{fig:2bend-bottomtop-left}
  \end{figure}
  
  \subcase{sc:bottomtop-left} $(u,w)$ is an L-shape and~$w$ lies to the left of~$u$;
  see Fig.~\ref{fig:2bend-bottomtop-left}.
  Assume first that~$w$ lies below~$v$.
  We claim that there is no vertex in the region bounded by the vertical segment
  of~$(u,w)$ from the left, the first horizontal segment of~$(u,v)$ from the bottom,
  the vertical segment of~$(u,v)$ to the right and the $y$-coordinate of~$w$
  from the top. Assume to the contrary that there is some vertex in this
  region and let~$x$ be the bottom-most one. 
  Every vertex has at least one incoming edge, and by invariant (I3) the target 
  vertex is not below the source vertex. Hence, there has
  to be an edge from some vertex~$y$ to vertex~$x$ such that~$y$ does not
  lie above~$x$. Since the E and the N port of~$u$ are already used, 
  $y$ cannot be~$u$. If~$y$ lies below~$x$,  by choice of~$x$ 
  the edge~$(y,x)$ has to intersect $(u,w)$ or $(u,v)$,
  which contradicts invariant (I1).
  Otherwise, $y$ lies next to~$x$ and $(y,x)$ is a horizontal segment.
  If~$y$ lies to the right (left) of~$x$, we choose the rightmost (leftmost) vertex~$z$
  such that there is a directed path from~$z$ to~$x$ that only contains horizontally
  drawn edges. If~$z$ lies outside the region,  some edge on this path has
  to cross $(u,v)$ or $(u,w)$, which contradicts (I1). Otherwise, we repeat
  the argument with~$z$; since~$z$ cannot have an incoming edge from a vertex
  with the same $y$-coordinate, it cannot have any incoming edge without 
  a crossing.
  
  Hence, this region is empty and we can move~$x$ upwards to the same $y$-coordinate
  as~$w$. Now $(u,w)$ uses the W port at~$u$ and we can use Case~\ref{c:bottomleft}
  to make $(u,v)$ an L-shape.
  
  On the other hand, if~$w$ does not lie below~$v$,  we can use the same
  argument that the area described above, but bounded form the top by the
  $y$-coordinate of~$v$, is empty. However, there has to be an edge that uses
  the S port of~$v$, and it has to be $y$-monotone by invariant (I3),
  so its source has to lie below~$v$; a contradiction.
  
  \begin{figure}[t]
    \centering
    \subcaptionbox{}{\includegraphics[page=10]{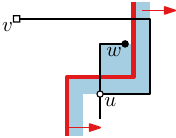}}
    \hfill
    \subcaptionbox{}{\includegraphics[page=11]{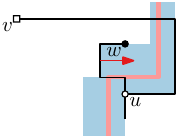}}
    \hfill
    \subcaptionbox{}{\includegraphics[page=12]{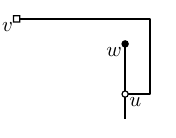}}
    \caption{Proof of Lemma~\ref{lem:2bend-bicon}, Case~\ref{sc:bottomtop-right}}
    \label{fig:2bend-bottomtop-right}
  \end{figure}
  
  \subcase{sc:bottomtop-right} $(u,w)$ is an L-shape and~$w$ lies to the right of~$u$;
  see Fig.~\ref{fig:2bend-bottomtop-right}.
  By the same argument as in Case~\ref{sc:bottomtop-left},~$w$ has to lie below~$v$.
  We can also use the exact argument to show that the region between $(u,w)$
  and $(u,v)$ is empty.
  
  We create a curve~$S$ as in Case~\ref{sc:bottomtop-up}.
  We stretch the drawing by moving~$S$ to the right such that~$u$ is
  placed directly below~$w$. Because of the empty region, we can now make $(u,w)$
  a vertical edge and then use Case~\ref{sc:bottomtop-up} to make $(u,v)$
  an L-shape.
  
  \begin{figure}[b]
    \centering
    \subcaptionbox{}{\includegraphics[page=13]{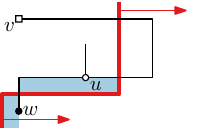}}
    \hfill
    \subcaptionbox{}{\includegraphics[page=14]{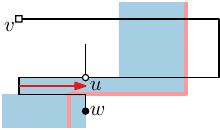}}
    \hfill
    \subcaptionbox{}{\includegraphics[page=15]{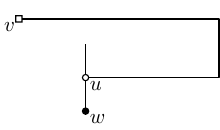}}
    \caption{Proof of Lemma~\ref{lem:2bend-bicon}, Case~\ref{sc:topleft-down}}
    \label{fig:2bend-topleft-down}
  \end{figure}
  
  \ccase{c:topleft} The N port and the W port at~$u$ are used. 
  Let~$(u,w)$ be the edge that uses the W port at~$u$; since $(u,v)$ is an
  outgoing edge and the edge at the N port has to be outgoing by invariant (I3),
  $(u,w)$ is an incoming edge at~$u$. By invariants (I3) and (I5), 
  it has to be drawn as a horizontal
  segment or as an L-shape such that~$w$ lies below~$u$.
  We distinguish two more cases based on the drawing style of $(u,w)$.
  
  \subcase{sc:topleft-down}  $(u,w)$ is an L-shape and~$w$ lies below~$u$;
  see Fig.~\ref{fig:2bend-topleft-down}. We create a curve~$S$ as follows: We
  start at some point~$p$ slightly to the top left of~$w$ and extend it downward 
  to infinity. Then we extend it from~$p$ to the right until it passes~$u$ and
  extend it upwards to infinity. 
  We place the
  curve close enough to $(u,w)$ such that no vertex or bend point
  lies between~$S$ and $(u,w)$.
  Then, we stretch the drawing by moving~$S$ to the right such that~$w$ is
  placed below~$u$. After this operation, the S port of~$u$ is free and 
  there is no edge or vertex on the vertical segment between~$u$ and~$w$,
  so we can make $(u,w)$ a vertical edge and then use Case~\ref{c:bottomtop}
  to make $(u,v)$ an L-shape.
  
  \begin{figure}[t]
    \centering
    \subcaptionbox{}{\includegraphics[page=16]{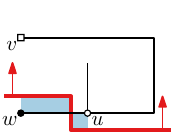}}
    \hfill
    \subcaptionbox{}{\includegraphics[page=17]{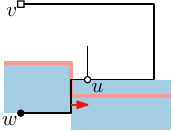}}
    \hfill
    \subcaptionbox{}{\includegraphics[page=18]{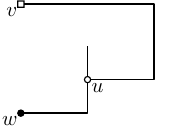}}
    \caption{Proof of Lemma~\ref{lem:2bend-bicon}, Case~\ref{sc:topleft-left}}
    \label{fig:2bend-topleft-left}
  \end{figure}
  
  \subcase{sc:topleft-left}  $(u,w)$ is a horizontal edge
  and~$w$ is at the same $y$-coordinate as~$u$;
  see Fig.~\ref{fig:2bend-topleft-left}. We now create an $x$-monotone curve~$S$ 
  as follows: We start at some point~$p$ slightly to the top left of~$u$ and 
  extend it leftward to infinity. Then we extend it from~$p$ to the bottom until it passes~$(u,w)$ and
  extend it rightwards to infinity. We place the
  curve close enough to $(u,w)$ such that no vertex or bend point
  lies between~$S$ and $(u,w)$.
  Then, we stretch the drawing by moving~$S$ upwards for a short distance.
  After this operation, the S port of~$u$ is free and 
  the whole region between~$w$ and~$u$ is empty,
  so we can make $(u,w)$ an L-shape and then use Case~\ref{c:bottomtop}
  to make $(u,v)$ an L-shape.
  
  Obviously, each of the above operations maintains all the invariants.
  Hence, by repeating them for every C-shape, we obtain the desired drawing
  of~$G_i$. 
\end{proof}
\end{backInTime}

\section{Omitted proofs from Section~\ref{sec:lower}}\label{app:lower}

The following construction is the same as the one of Hong et 
al.~\cite{DBLP:conf/cocoon/HongELP12} to prove an exponential area lower
bound for straight-line drawings of 1-plane graphs, with
two edges added to make the graph biconnected.

\begin{backInTime}{lowerStraight}
\begin{theorem}
  \lowerStraightTwoRegText
\end{theorem}
\begin{figure}[t]
  \subcaptionbox{The graph $G_5$ \label{fig:lowerbound-2reg-a}}{\includegraphics[page=1]{lowerbound-2reg}}
  \hfill
  \subcaptionbox{A straight-line drawing of~$G_3$}{\includegraphics[page=2]{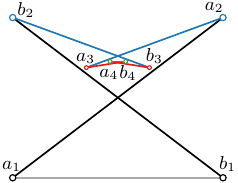}}
  \hfill
   \subcaptionbox{The tower of the $W_i$s (schematic)\label{fig:lowerbound-2reg-c}}{\includegraphics[page=3]{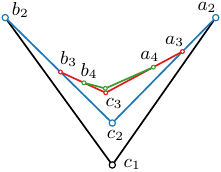}}
  \caption{The construction for Theorem~\ref{thm:lower-straight-2reg}}
  \label{fig:lowerbound-2reg}%
\end{figure}
\begin{proof}
  Let $G_k$ be the plane graph given by the cycle  $a_1\ldots,a_{k+1},b_{k+1},\ldots,b_1,a_1$
  and the embedding shown in Fig.~\ref{fig:lowerbound-2reg-a}. 
  We denote the crossing 
  between two edges $a_ia_{i+1}$ and $b_ib_{i+1}$ with $c_i$. Let $W_i$ be the 
  quadrilateral with vertices $c_{i-1}$, $a_i$, $c_i$ and $b_i$. Note that the 
  edges in $W_i$ incident to $c_i$ and the edges in $W_{i+1}$ incident to $c_i$
  have the same (pair of) slopes. Hence, we can rotate $W_{i+1}$ by $\pi$ around $c_i$ and align it with $W_i$ such 
  that the edges meeting in $c_i$ of both quadrilaterals overlap. The rotation will maintain the slopes.
  Thus we can draw a \enquote{tower} of disjoint copies of all $W_i$ (see Fig.~\ref{fig:lowerbound-2reg-c}). 
  For all $1<i<k$ 
  the supporting lines of $a_ic_i$ and $a_{i+1}c_{i+1}$ differ by rotation (angle $<\pi$) in the same direction.
  The total rotation of those edges cannot exceed $\pi$, since no $a_ic_{i}$ can \enquote{overtake} $b_2a_2$.
  As a consequence, the slopes of all edges $a_ic_i$ are different and thus also all the slopes of the edges $a_ia_{i+1}$
  have to be different.
\end{proof}

\begin{lemma}
	\lowerStraightThreeRegText
\end{lemma}
\begin{figure}[b]
  \subcaptionbox{The graph $G$  \label{fig:lowerbound-3reg-a}}{\includegraphics[page=1,scale=.9]{lowerbound-3reg}}
  \hfill
  \subcaptionbox{A straight-line drawing of~$G$  \label{fig:lowerbound-3reg-b}}{\includegraphics[page=2,scale=.9]{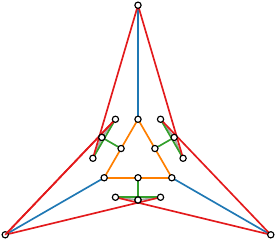}}
    \hfill
  \subcaptionbox{Rearrangement of~$G$ (schematic)  \label{fig:lowerbound-3reg-c}}{\includegraphics[page=4,scale=.9]{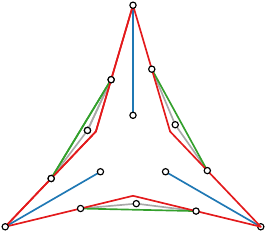}}
  \caption{The construction for Lemma~\ref{lem:lower-straight-3reg}}
  \label{fig:lowerbound-3reg}
\end{figure}
\begin{proof}
	Consider the graph~$G$ depicted in Fig.~\ref{fig:lowerbound-3reg}a--b. To simplify the analysis
	we exploit a similar idea as in Theorem~\ref{thm:lower-straight-2reg}. Let $x_i$ be the 
	crossing between $a_ic_i$ and $a_{i+1}e_{i}$ (indices modulo 3).
	Fix any straight-line drawing of~$G$ and let $T_i$ be the triangle $e_ic_ix_i$ including 
	the two segments~$e_id_i$ and $c_id_i$. For $i=1,2,3$ we cut $T_i$, rotate it by $\pi$ around $x_i$ 
	and put it back to the drawing. This leaves the slopes unchanged. By this we obtain a drawing
	which contains a pseudo-triangle, whose chains ($a_ic_id_ie_ia_{i+1}$) 
	have 4 edges (see Fig.~\ref{fig:lowerbound-3reg-c}). Further, for every chain there is
	an edge ($e_ic_i$) between the second and fourth vertex cutting off~$d_i$. 
	The edges of a pseudo-triangle have different slopes. Thus, we have 12 different slopes here. 
	Moreover, if you traverse the edges of a pseudo-triangle in cyclic order they will be 
	ordered by slope. Since $a_ib_i$ 
	is sandwiched  between $a_ic_i$ and $a_ie_{i+2}$ we have three more distinct slopes. Finally, we note
	that replacing the edges $e_id_i$ and $c_id_i$ with $e_ic_i$ gives another pseudo triangle that avoids
	the slopes of $e_id_i$ and $c_id_i$. As a consequence, the edges $e_ic_i$ will give us 
	three new slopes and we end up with 18 different slopes.
  
  To obtain an infinite family of graphs, observe that we do not use the edges
  between~$b$- and~$f$-vertices in our analysis. Hence, we can subdivide
  the edges $(b_1,f_3)$ and $(b_2,f_1)$ several times and connect pairs of
  subdivision vertices.
\end{proof}

\begin{theorem}
	\lowerStraightThreeGenText
\end{theorem}
\begin{figure}[tb]
  \centering
  \includegraphics[page=3]{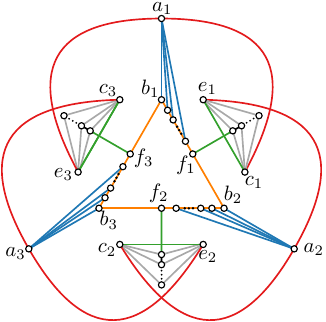}
  \caption{The construction for Theorem~\ref{thm:lower-straight-gen}}
  \label{fig:lowerbound-gen}
\end{figure}
\begin{proof}
 Consider the graph depicted in Fig.~\ref{fig:lowerbound-gen}. 
 The degree of $a_i,c_i$ and $e_i$ is $\Delta$.
 
 We can repeat the argument of the proof of Lemma~\ref{lem:lower-straight-3reg}. 
 There are only two differences: (i) instead of a single edge $a_ib_i$ there is a bundle 
 of edges incident to $a_i$. However the whole bundle lies in between $a_ic_i$ and $a_ie_{i+2}$
 and therefore the slopes of these edges are distinct. (ii) Instead of the pseudo-triangle
 with chains $a_ic_id_ie_ia_{i+1}$ we have now a sequence of nested chains given by the
 edges incident to $e_i$ and $c_i$. All these \enquote{subchains} are contained in the triangle $e_ic_ix_i$,
 where $x_i$ is the crossing between $a_ic_i$ and $a_{i+1}e_{i}$. Their slopes lie between the
 slopes of $a_ic_i$ and $e_ia_{i+1}$. This means that the three edge-bundles of subchains are separated
 by slopes. Clearly the slopes within each bundle have to be different. 
  
 Counting the slopes we see that we have the 18 slopes of the subgraph shown in Fig.~\ref{fig:lowerbound-3reg-a}
 and then there are 9 vertices, each incident to $\Delta-3$ new edges, that will need a new slope. In total
 we need $18 + 9 (\Delta -3) = 9 (\Delta -1)$ distinct slopes.

 To obtain an infinite family of graphs, we can again subdivide 
 $(b_1,f_3)$ and $(b_2,f_1)$ several times and connect pairs of the
 subdivision vertices.
 \end{proof}
\end{backInTime}

\end{document}